\documentclass[a4paper,onecolumn,accepted=2023-03-27]{quantumarticle}
\pdfoutput=1
\usepackage[utf8]{inputenc}
\usepackage[english]{babel}
\usepackage[T1]{fontenc}

\usepackage{graphicx}
\usepackage{mathtools,amssymb,amsmath, amsthm,stmaryrd,thmtools,braket}
\usepackage[numbers,sort&compress]{natbib}
\usepackage{xr-hyper}
\usepackage[colorlinks=true, urlcolor=blue,citecolor=blue,anchorcolor=blue]{hyperref}
\usepackage{algorithm}
\usepackage[noend]{algpseudocode}

\usepackage{bm}
\newcommand{\expect}{{\rm I\kern-.3em E}}

\DeclarePairedDelimiter{\ceil}{\lceil}{\rceil}
\DeclarePairedDelimiter{\floor}{\lfloor}{\rfloor}

\declaretheorem[
numberwithin = section,
]{definition}

\declaretheorem[
numberwithin=section
]{fact}

\begin{document}
\title{Parallelization techniques for quantum simulation of fermionic systems}
\author{Jacob Bringewatt}
\affiliation{Department of Physics, University of Maryland, College Park, Maryland 20742, USA}
\affiliation{Joint Center for Quantum Information and Computer Science, NIST/University of Maryland, College Park, Maryland 20742, USA}
\affiliation{Joint Quantum Institute, NIST/University of Maryland, College Park, Maryland 20742, USA}

\author{Zohreh Davoudi}
\affiliation{Department of Physics, University of Maryland, College Park, Maryland 20742, USA}
\affiliation{Maryland Center for Fundamental Physics, University of Maryland, College Park, Maryland 20742, USA}
\affiliation{Institute for Robust Quantum Simulation, University of Maryland, College Park, Maryland 20742, USA}

\begin{abstract}
\noindent
   Mapping fermionic operators to qubit operators is an essential step for simulating fermionic systems on a quantum computer. We investigate how the choice of such a mapping interacts with the underlying qubit connectivity of the quantum processor to enable (or impede) parallelization of the resulting Hamiltonian-simulation algorithm. It is shown that this problem can be mapped to a path coloring problem on a graph constructed from the particular choice of encoding fermions onto qubits and the fermionic interactions onto paths. The basic version of this problem is called the weak coloring problem. Taking into account the fine-grained details of the mapping yields what is called the strong coloring problem, which leads to improved parallelization performance. A variety of illustrative analytical and numerical examples are presented to demonstrate the amount of improvement for both weak and strong coloring-based parallelizations. Our results are particularly important for implementation on near-term quantum processors where minimizing circuit depth is necessary for algorithmic feasibility.
   \vspace{0.3 cm}   
   \noindent
   {Preprint Report No.} UMD-PP-022-06
\end{abstract}

\maketitle

\section{Introduction}
The ability to simulate complex fermionic systems is an important area of promise for quantum computers with applications ranging from quantum chemistry and condensed matter physics to nuclear and high-energy physics~\cite{georgescu2014quantum,bauer2020quantum,davoudi2022quantum,NSAC-QIS-2019-QuantumInformationScience}. Before performing such a simulation, however, one must map from the fermionic operators to operators acting on the Hilbert space of the qubits of the quantum computer.
A common approach to performing such a mapping is to use the Jordan-Wigner transformation~\cite{jordan1928paulische, ortiz2001quantum}, which encodes local fermionic operators on $N$ fermionic modes as non-local qubit operators on $N$ qubits. This non-locality, which manifests as strings of Pauli-$Z$ operators, is the price of obtaining the correct fermionic anti-commutation relations when using qubit operators. Unfortunately, even for physically-local fermionic interactions in higher than 1+1 dimensions, the length of these Pauli-$Z$ strings  can scale with the system size. This results in costly fermionic simulations~\cite{havlicek2017operator, clinton2021hamiltonian} on near-term quantum devices where the two-qubit entangling-gate (e.g., CNOT) count of an algorithm is expected to be the limiting factor. In particular, the Pauli weight of an operator $G$ (the number of qubits on which it acts non-trivially) is directly related to the number of two-qubit entangling gates needed to implement the unitary $U=\mathrm{exp}(-iG)$~\cite{mikeandike}. Nonetheless, recent progress has resulted in improvement in both entangling-gate count and in circuit depth when simulating given Jordan-Wigner-transformed fermionic Hamiltonians using product formulas, resorting to e.g., suitable term ordering and nesting strategies~\cite{hastings2014improving} or fermionic SWAP networks~\cite{babbush2017low, kivlichan2018quantum, cade2020strategies}.

A number of other mappings from fermions to qubits have been proposed in the literature~\cite{bravyi2002fermionic, verstraete2005mapping, whitfield2016local, steudtner2017lowering, steudtner2019quantum,setia2019superfast, jiang2019majorana, chen2020exact, chien2020custom, derby2021compact, kirby2021second, chiew2021optimal}. Many of these proposals aim to map local fermionic operators to local qubit operators, forming a class that is called local encodings in this paper. Local encodings trade operator non-locality for state non-locality as a vehicle for encoding fermionic anti-commutation relations in qubits. In particular, one finds that to preserve the appropriate anti-commutation relations via a local encoding, one must: a) increase the number of qubits, and b) restrict the state of the system to lie within some subspace of the Hilbert space---typically the logical codespace a (modified) toric code. At the price of these complications, one generally obtains lower gate counts required for simulation. The comparative analyses of various encodings given the Hamiltonian under study, the quantum resources to be optimized, and the architecture connectivity constitute an active area of research, see e.g., Refs.~\cite{cade2020strategies, derby2021compact, tranter2015b,tranter2018comparison, steudtner2018fermion}. 

In this work, we explore the potential for parallelization (that is the ability to simultaneously simulate several Hamiltonian terms) in local encodings, hence reducing the circuit depth of the simulation. In this context, product formula-based Hamiltonian-simulation algorithms based on Trotter-Suzuki decomposition of the time-evolution operator~\cite{suzuki1991general, wiebe2010higher, childs2021theory} are best suited for this analysis, nonetheless, other simulation algorithms~\cite{childs2012hamiltonian, berry2015simulating, low2017optimal, low2019hamiltonian, chakraborty2018power, gilyen2019quantum, kalev2021quantum, rajput2021hybridized} can also benefit from the parallelization strategy explored here. We consider the parallelization problem in connection to (a slightly abstracted version of) the underlying qubit architecture of the quantum computer, and emphasize, both analytically and numerically, connections between qubit architectures, fermionic-encoding locality, and parallelization. It is found that the problem of parallelization is equivalent to path coloring on a graph that represents the particular fermion-to-qubit mapping under consideration and the physical interactions among the fermionic modes. Consequently, heuristic classical algorithms can be used to inform efficient implementations of fermionic simulations on quantum hardware.

The particular graph-theoretic approach of this work is enabled by the strategy undertaken in Ref.~\cite{chien2020custom}, in which a general framework for local fermionic encodings of the sort described above is developed. In particular, it was demonstrated how to disconnect the interaction graph of the fermionic modes being simulated and the so-called \emph{system graph}, which determines the fermionic encoding in a flexible and qubit-architecture-aware manner. 
This separation enables the construction of the so-called \emph{custom fermionic codes}, which are a generalization of the Bravyi and Kiteav superfast encoding~\cite{bravyi2002fermionic, setia2019superfast}. The Jordan-Wigner transformation is a limiting case of such custom codes. The degree of non-locality can be reduced upon introducing further qubits and local connectivity on the system graph at will, and such choices amount to a range of custom encodings. The input to our parallelization problem is such a system graph, which fixes the encoding chosen to implement the interactions in the original fermionic Hamiltonian. The question investigated in this paper is to what extent the Hamiltonian simulation can be parallelized, and whether certain system graphs are best suited for maximal parallelizability.

The structure of this paper is as follows. Custom fermionic codes of Ref.~\cite{chien2020custom} is reviewed in Sec.~\ref{s:customfermioniccodes}. In Sec.~\ref{s:parallelization}, we demonstrate how the problem of parallelizing product formula-based Hamiltonian-simulation algorithms maps to path coloring on the system graph, and is, therefore, a NP-hard problem. This is named the weak coloring problem. By considering the fine-grained details of the fermion-to-qubit mapping below the abstraction level of the system graph, another path coloring problem is defined. This is called the strong coloring problem. Analytic results on the weak and strong coloring problems for a few illustrative systems are presented in Sec.~\ref{s:analyticresults}. A numerical approach to heuristically solving the weak and strong coloring problems is presented Sec.~\ref{s:numericalresults}. The numerical and analytic results are then compared for these system graphs, exhibiting consistency. We further numerically investigate the weak and strong coloring problems for a few realistic system graphs designed for current qubit architectures. It is found that by solving (or finding heuristics for) the more detailed strong coloring problem, one can often obtain significant gains in parallelizability compared to the more abstracted weak coloring problem. For most system graphs, these improvements are a constant factor---for instance, in the case of a star system graph or complete system graph, strong coloring asymptotically provides up to a factor of two improvement over weak coloring. However, we also construct an example for which the advantage gained grows linearly in the system size, which is the maximal possible gain from considering strong coloring. Both weak and strong coloring approaches provide large reductions in circuit depth compared to a naive sequential approach. Finally, the performance gains of strong coloring depend heavily on the choice of enumerating qubits in the encoding---a feature that is also taken advantage of in Ref.~\cite{chiew2021optimal} to provide optimal fermion-qubit mappings. Sec.~\ref{s:conclusion} includes a summary of the results and a discussion of possible directions for future study. The code generating the colored graph from the system graph and the corresponding physical interactions is provided as supplemental material~\cite{github}. 

\section{Custom Fermion-to-Qubit Mappings}\label{s:customfermioniccodes}
\subsection{Setup}
Consider a system of $N$ fermionic modes. A general
fermionic Hamiltonian may be written in the second quantization as 
\begin{equation}\label{eq:fermionicH}
    H=\sum_{uv} \kappa_{uv} a_u^\dagger a_v + \sum_{uvwx} \kappa_{uvwx} a_u^\dagger a_v^\dagger a_w a_x+ \cdots,
\end{equation}
where $a_u^\dagger$ and $a_u$ are the fermionic creation and annihilation operators on site $u$, respectively, satisfying the standard fermionic anticommutation relations
\begin{align}\label{eq:anticommrelations}
    \{a_u, a_v^\dagger\}&=\delta_{uv},\quad
    \{a_u, a_v\}=\{a_u^\dagger, a_v^\dagger\}=0,
\end{align}
and $\kappa_{uv}, \kappa_{uvwx}$ are some coefficients consistent with the Hermiticity of the Hamiltonian. It is convenient to consider a Majorana basis $\gamma, \gamma'$ for the fermionic operators as
\begin{align}\label{eq:majorana}
    \gamma_u&=a_u^\dagger+a_u, \quad  
    \nonumber\\
    \gamma_u'&=i(a_u^\dagger-a_u),
\end{align}
that clearly satisfy
\begin{align}
    \{\gamma_u,\gamma_v\} = 2\delta_{uv}, \quad \{\gamma_u',\gamma_v'&\} = 2\delta_{uv},\quad \{\gamma_u,\gamma_v'\} = 0.
\end{align}
To simulate a fermionic Hamiltonian on a quantum computer, one must first map from fermions to qubits while preserving these anti-commutation relations. The standard approach is the Jordan-Wigner mapping from $N$ fermionic modes to $N$ qubits,
\begin{align}\label{eq:JW}
    &a_u\rightarrow \prod_{v<u} Z_v(X_u+iY_u), \nonumber\\
    &a_u^\dagger \rightarrow \prod_{v<u} Z_v(X_u-iY_u),
\end{align}
where $X_u, Y_u$, and $Z_u$ are Pauli operators on the $u$-th qubit. The Jordan-Wigner transformation requires non-local qubit operations whose weight scales with the system size. These high Pauli-weight operators directly translate to increased gate counts for quantum simulation, and has stimulated various strategies to alleviate the simulation cost when resorting to a Jordan-Wigner mapping~\cite{hastings2014improving, babbush2017low, kivlichan2018quantum, cade2020strategies}.

\subsection{Local Fermion-to-Qubit Mappings}
There are many other approaches to mapping from fermions to qubits which aim to address the shortcomings of the Jordan-Wigner transformation. For instance, the Bravyi-Kitaev transformation encodes both occupation information (like the Jordan-Wigner transformation) and parity information in such a way that single fermionic operators act non-trivially on at most $\mathcal{O}(\log_2 N)$ qubits~\cite{bravyi2002fermionic}. This is in contrast to the linear scaling of the Pauli weight of qubit operators in system size for the Jordan-Wigner transformation. A simpler ternary-tree-based mapping from $N$ fermionic modes to $N$ qubits performs even better, leading to provably minimal Pauli weights in the average case. In this case, a single fermionic operator acts on $\ceil{\log_3(2N+1)}$ qubits~\cite{jiang2020optimal}. One can think of such a mapping as a generalization of the Jordan-Wigner transformation from a 1D chain to tree graphs~\cite{vlasov2019clifford}.

Fully local encodings---in the sense that local fermionic operators map to local qubit operators--- are possible with the addition of ancilla qubits. An important example is the Bravyi and Kitaev superfast encoding~\cite{bravyi2002fermionic} and its generalizations~\cite{setia2019superfast, chien2020custom}. A multitude of other local mappings have also been developed, often aimed at minimizing the qubits required, while still maintaining local, low Pauli-weight qubit operators \cite{ball2005fermions, whitfield2016local, steudtner2017lowering, steudtner2019quantum, chen2020exact,  derby2021compact, jiang2019majorana, derby2021compact2, kirby2021second}. These local mappings can generally be understood as generalizations of the toric code~\cite{kitaev2003fault}. In particular, all known local fermionic encodings are equivalent to the toric code on some lattice up to deformations by a constant-depth circuit of local Cliffords~\cite{chien2020custom, derby2021compact}. Fermionic-pair excitations in the local encoding arise as freely deformable strings of Pauli operators on the lattice, and the fermionic anti-commutation relations are enforced via restriction to a particular code subspace of the ancilla-extended Hilbert space. Equivalently, one could view the ancilla qubits as being used to couple to an auxillary gauge field~\cite{chen2020exact, chen2018exact}. A given local mapping, therefore, corresponds to a particular ``gauge theory'' and restricting the simulation to a particular subspace is equivalent to the choice of Gauss's law sector in the corresponding gauge theory. In either view, observe that local fermionic encodings of this sort require the preparation of a toric-code state. Therefore, local fermion-to-qubit mappings trade non-locality in the operators for extra qubits and non-locality in the states. To ensure the simulation proceeds in the allowed subspace of the Hilbert space---that is, that the local and non-local constraints are satisfied--- strategies similar to preserving (gauge) symmetries in lattice-gauge-theory simulations~\cite{stryker2019oracles,tran2021faster,lamm2020suppressing,halimeh2021gauge,van2021suppressing,stannigel2014constrained,kasper2020non,nguyen2021digital} can be explored.

\subsection{Custom Fermionic Codes}
This work focuses on a particular class of fermion-to-qubit encodings---the so-called custom fermionic encodings---developed by Chien and Whitfield~\cite{chien2020custom} as a generalization of the construction by Setia et al.~\cite{setia2018bravyi}. These mappings are, in turn, a generalization of the Bravyi and Kitaev superfast encoding~\cite{bravyi2002fermionic}. For our purposes, the essential feature of custom fermionic codes is that they allow for a variety of different encodings ranging from local to quasi-local to highly non-local ones. This flexibility permits trading resources like the number of qubits, qubit connectivity, and Pauli weight of simulated operators in an architecture-aware manner. In this paper, we will add the parallelizability of the resulting Hamiltonian-simulation algorithm to this list. 

Let us briefly review how to construct a custom fermionic code. One can introduce edge operators $A_{uv}$ and vertex operators $B_u$ which are defined as
\begin{align}
    A_{uv}&=-i\gamma_u\gamma_v, \label{eq:A}\\
    B_{u}&=-i\gamma_u\gamma'_u, \label{eq:B}.
\end{align}
These operators suffice to generate all parity-preserving fermionic operators in a Hamiltonian of the form Eq.~(\ref{eq:fermionicH}). Therefore, the Hamiltonian with $N$ fermionic modes
\begin{equation}\label{eq:H}
    H_K=\sum_{uv} \kappa_{uv} a_u^\dagger a_v,
\end{equation}
for some symmetric, real constants $\kappa_{uv}=\kappa_{vu}$, can be expressed as
\begin{align}\label{eq:alltoall_edge}
    H_K&=-\frac{i}{2}\sum_{u< v} \kappa_{uv} (A_{uv}B_v+B_uA_{uv}) -\frac{1}{2}\sum_{u} \kappa_{uu} B_u\\
    &=-\frac{i}{2}\sum_{u< v} \kappa_{uv}A_{uv}(B_v-B_u) -\frac{1}{2}\sum_{u} \kappa_{uu} B_u,
\end{align}
up to constant terms that can be ignored.

The interaction set $\mathcal{T}$ can now be defined as the set of all terms with non-zero coefficients in the re-expressed Hamiltonian. Furthermore, an interaction graph $\Gamma=\{V_\Gamma, E_\Gamma\}$ can be defined with vertices corresponding to each fermionic mode and an edge joining any pair of vertices $(u,v)$ such that the edge operator $A_{uv}B_v$ belongs to $\mathcal{T}$. For instance, for the Hamiltonian in Eq.~(\ref{eq:alltoall_edge}) with $\kappa_{uv}\neq 0$, the interaction set is
\begin{equation}\label{eq:interactionset}
    \mathcal{T}=\{A_{uv}B_v\}_{u\neq v}\cup \{B_u\}_{u\in \{1,\cdots N\}},
\end{equation}
and the corresponding interaction graph is a complete graph $K_N$ on $N$ vertices. In what follows, it is assumed without loss of generality that $\Gamma$ is connected, as if $\Gamma$ is disconnected, one is dealing with two physically independent systems, and can consider the connected case on each system separately. 

From Eq.~(\ref{eq:majorana}) and Eqs.~(\ref{eq:A}) and (\ref{eq:B}), one can show that the edge and vertex operators obey the following relations
\begin{eqnarray}\label{eq:abrelations}
    &&B_{u}^\dagger=B_u, \hspace{5.25 cm} A_{uv}^\dagger=A_{uv}, \nonumber\\
    &&B_u^2=A_{uv}^2=1,   \hspace{4.5 cm}  [B_u, B_v]=0, \nonumber \\
    &&A_{uv}=-A_{vu},     \hspace{4.75 cm}  A_{uv}B_w=(-1)^{\delta_{uw}+\delta_{vw}}B_wA_{uv}, \nonumber \\
    %\omit\rlap{\qquad$
    &&A_{uv}A_{wx}=(-1)^{\delta_{uw}+\delta_{ux}+\delta_{vw}+\delta_{vx}}A_{wx}A_{uv},
    %$,}
    \nonumber \\
    %\omit\rlap{\qquad $
    &&i^{|C|}\prod\limits_{\nu=1}^{|C|} A_{c_\nu c_{\nu+1}} =I,
    %$,}
\end{eqnarray}
where in the final equality $C$ is any cycle in $\Gamma$ specified via an ordered list of vertices $C=\{c_1, c_2, ... , c_{|C|}, c_1\equiv c_{|C|+1}\}$ for $c_\nu\in V_\Gamma$ with only the final vertex repeated, and, therefore, the product is over all edge operators in the cycle. 

Next, a second graph can be introduced, the so-called system graph $\Sigma=\{V_\Sigma, E_\Sigma\}$. A valid system graph is any undirected, connected graph with vertex set $V_\Sigma=V_{\mathrm{phys}}\cup V_{\mathrm{virt}}$ equipped with a bijective mapping $M:V_\Gamma\rightarrow V_\mathrm{phys}$---that is, $|V_\Gamma|=|V_\mathrm{phys}|$. The subscripts are shorthand for physical vertices and virtual vertices. The physical vertices correspond to physical fermionic modes in the interaction graph and the virtual vertices (if they exist) correspond to additional auxillary fermionic modes that can be freely introduced. 
One can envision constructing the system graph from the interaction graph by adding an arbitrary number of virtual vertices and adding or removing any edges so long as the final graph is connected. Note that the condition that the graph is connected implies that any two vertices that were connected before are still connected via physical or auxillary vertices. This connectivity condition is sufficient for one to implement any interaction terms in the Hamiltonian, see e.g., Eq.~(\ref{eq:paths}) below. Consequently, the edge set of the system graph can be completely different from that of the interaction graph.

An encoding of a fermionic system on such a graph $\Sigma$ consists of associating with each vertex $v\in V_\Sigma$ a set of
\begin{equation}
n_v:=\ceil{d(v)/2}
\end{equation}
qubits, where $d(v)$ is the degree of vertex $v$. A set of $2n_v$ Pauli operators $\{{\tilde{\gamma}}_v^1, {\tilde{\gamma}}_v^2, \cdots, {\tilde{\gamma}}_v^{2n_v}\}$ can then be defined on these qubits. In the following, these operators are referred to as local Majoranas. Note that we have introduced the convention of using subscripts $\{u,v,w,\cdots\}$ to index vertices of $\Sigma$ and superscripts $\{i,j,k,\cdots\}$ to index local quantities such as enumerations of the local Majoranas or indices of internal qubits. Furthermore, subscripts $\{\nu, \mu, \cdots\}$ are used in various places for indexing generic sets. 

The local Majoranas can be any choice of operators that satisfy the following conditions:
\begin{enumerate}
    \item They obey the Majorana-operator\footnote{Note that these correspond to both types of $\gamma$ and $\gamma'$ operators defined in Eq.~(\ref{eq:majorana}).} properties including anti-commutation relations with other local Majoranas defined on the vertex. Furthermore, they must commute with the local Majoranas on other vertices. That is,
\begin{align}
       \tilde{\gamma}_v^{k\dagger}=\tilde{\gamma}_v^k, \quad
       \{{\tilde{\gamma}}_v^j, {\tilde{\gamma}}_v^k\}=2\delta_{jk}, \quad [{\tilde{\gamma}}_u^j, {\tilde{\gamma}}_v^k]=0 \,\, \mathrm{for} \,\, u\neq v.
    \end{align}
    \item They generate the full Pauli group on the $n_v$ qubits associated with $v\in V_\Sigma$.
\end{enumerate}

Any explicit choice for the local Majoranas can be mapped to any other via a Clifford circuit acting on the qubits associated with that vertex~\cite{chien2020custom}. Most simply, one could consider encoding the local Majoranas via a Jordan-Wigner transformation. That is, given some enumeration of the qubits in a vertex $v$, one has
\begin{equation}\label{eq:JWencoding}
       \{{\tilde{\gamma}}_v^1, {\tilde{\gamma}}_v^2, {\tilde{\gamma}}_v^3, {\tilde{\gamma}}_v^4, {\tilde{\gamma}}_v^5, {\tilde{\gamma}}_v^6, \cdots\} \longrightarrow \{X_{v}^{1}, Y_{v}^{1}, Z_{v}^1X_{v}^2, Z_{v}^1Y_{v}^2, Z_{v}^1Z_{v}^2X_{v}^{3}, Z_{v}^1Z_{v}^2Y_{v}^3, \cdots\}.
\end{equation}
It is straightforward to verify that this choice satisfies the conditions above. One could also use other encodings---for instance,  Fenwick trees~\cite{whitfield2016local, setia2018bravyi} or ternary trees~\cite{jiang2020optimal}. This work concerns only the case of a Jordan-Wigner encoding of the local Majoranas. Note, however, that the same techniques and many of the qualitative results will apply similarly to these other choices.

Once the local Majoranas are specified, each local Majorana can be associated with an edge of that same vertex. That is, both the local Majoranas associated with a vertex and the edges connecting the vertex to its neighbors are enumerated in $\Sigma$. The $j$-th local Majorana corresponding to a vertex is then associated with the edge connecting it to its $j$-th neighbor. Therefore, each edge $e\in E_\Sigma$ has two associated local Majoranas---one at each endpoint. Given such a choice, encoded edge operators acting on qubits can be defined as
\begin{equation}\label{eq:qubitedgeop}
    \tilde{A}_{uv}=\epsilon_{uv} {\tilde{\gamma}}_u^{\xi_u(v)}{\tilde{\gamma}}_v^{\xi_v(u)},
\end{equation}
where $v$ is the $\xi_u(v)$-th neighbor of $u$, $u$ is the $\xi_v(u)$-th neighbor of $v$, and the Levi-Civita tensor $\epsilon_{uv}$ is defined with respect to an arbitrary choice of orientation for each edge in $\Sigma$. In particular, we let $\epsilon_{uv}=1$ if $u$ is the head of the oriented edge $(u,v)$ and $\epsilon_{uv}=-1$ if $u$ is the tail. This choice of enumerating the edges of each vertex $u$, as specified by picking $\xi_u(v)$ for each neighbor $v$ of $u$, will become important later in Sec.~\ref{ss:limits} when discussing the strong-coloring problem.

Furthermore, the vertex operator on qubits can be encoded as
\begin{equation}\label{eq:qubitvertop}
    \tilde{B}_u=(-i)^{n_u}\prod_{j=1}^{2n_u} {\tilde{\gamma}}_u^j.
\end{equation}

One can verify that the choices of encodings in Eqs.~(\ref{eq:qubitedgeop}) and (\ref{eq:qubitvertop}) satisfy all but the final loop condition in Eq.~(\ref{eq:abrelations}). To satisfy the loop condition, it is necessary to restrict the simulation to the subspace of the total Hilbert space that does satisfy this condition. In the context of quantum error correction, this space is the codespace stabilized by the loop operators $\tilde L$ around cycles $C$ on $\Gamma$, defined by
\begin{equation}
    \tilde L(C)=i^{|C|}\prod\limits_{j=1}^{|C|} \tilde A_{c_jc_{j+1}}.
\end{equation}
As the encoded edge and vertex operators commute with the loop operators, once a state is initialized in the code subspace, the simulation remains in that subspace assuming no algorithmic or experimental errors. For considerations regarding boundary conditions and fermionic parity, see Refs.~\cite{setia2019superfast, chien2020custom}. Here, we consider only open boundary conditions for simplicity, but other boundary conditions can be analyzed within the framework of this work as well.

Once in the code subspace, the mapping from fermionic edge and vertex operators can be performed to qubit edge and vertex operators, $A_{uv}\rightarrow \tilde  A_{uv}$ and $B_u\rightarrow\tilde B_u$. This completes the mapping from a fermionic Hamiltonian $H$ to a qubit Hamiltonian $\tilde{H}$,
\begin{equation}\label{eq:map}
    \mathcal{M}:H=\sum_{\nu} \kappa_\nu h_\nu \longrightarrow \tilde{H}=\sum_{\nu} \kappa_\nu \tilde{h}_\nu,
\end{equation}
where $\kappa_\nu$ are constants related to the original coupling coefficients and the $h_\nu$ and $\tilde{h}_\nu$ are products of edge and vertex operators on fermionic Majorana modes and qubits, respectively.

This mapping $\mathcal{M}$ is not unique---the exact form of each $\tilde h_\nu$ depends not only on the particular system graph $\Sigma$ that determines the mapping, but also on the paths through the system graph chosen to simulate each corresponding $\tilde h_\nu$, as well as the choice of encoding of the local Majoranas. In particular, one does not necessarily need to have a direct edge $(u,v)\in E_\Sigma$ to implement $\tilde A_{uv}$. An edge operator between two modes $u$ and $v$ not directly connected is given by a product of edge operators along any path $
P_{uv}=\{p_1=u, p_2, \cdots, p_{|P_{uv}|}=v \}$ connecting the two modes. That is,\footnote{Observe that Eq.~(20) slightly overloads the notation $\tilde{A}_{uv}$, as strictly speaking, the $\tilde{A}_{uv}$ operators on the left- and right-hand sides of the equation have a slightly different meaning. In particular, one should distinguish between $\tilde{A}_{uv}$ that, given the system graph, can be directly implemented as in Eq.~(\ref{eq:qubitedgeop}), and those that cannot and must be implemented via a product of such operators as in Eq.~(\ref{eq:paths}). The meaning should be clear from the context. Note that tilde operators always denote those acting on qubits and not on fermionic modes.}
\begin{align}\label{eq:paths}
    \tilde A_{uv}&=\prod_{\nu=1}^{|P_{uv}|-1}\tilde A_{p_\nu p_{\nu+1}}.
\end{align}
Therefore, given the assumption that $\Sigma$ is connected, all edge operators in the qubit Hamiltonian $\tilde H$ can be implemented by choosing any path between the relevant vertices. 
Again, these path choices are not unique. While the precise details depend on these choices, it is always true that each $\tilde h_\nu$ is a string of Pauli operators on qubits. That is $\tilde{h}_\nu=\{X, Y, Z, I\}^{\otimes{n}}$, where $n$ is the number of qubits $\tilde{h}_\nu$ acts on. 
Importantly, whether an operator $\tilde{A}_{uv}$ can be implemented directly or must be implemented via a path of such operators through the system graph, it obeys all the same relations given in Eq.~(\ref{eq:abrelations}).

These choices do matter, however. In particular, recall virtual vertices are allowed in the system graph which, at the cost of more qubits, enable more choices of paths between different physical vertices. This tradeoff between more qubits and more direct (and correspondingly, more local) paths for implementing the required Pauli operators is the essential tension in regards to optimizing a Hamiltonian-simulation algorithm for a fermionic system in this construction.

\subsection{Prior Work on Optimizing System Graphs}
Some of the tradeoffs implied by the custom fermionic encoding have already been explored. In particular, Ref.~\cite{chen2020exact} discusses how the flexible framework of custom fermionic codes allows for designing fermionic encodings suited to particular qubit architectures by balancing the number of qubits required for an encoding with the Pauli weight of the resulting operators. In one limit, where the system graph is a line graph, one recovers the Jordan-Wigner transformation. By adding qubits and connectivity in the system graph, one can reduce the Pauli weight of the resulting operators, obtaining local or quasi-local encodings. This exact tradeoff was explored in detail for a variety of different system graphs in Ref.~\cite{chien2020custom} for the 2-body SYK model, which has all-to-all coupled fermions. 

Observe that the tradeoff between Pauli weight of operators and numbers of qubits and qubit connectivity is directly related to the properties of the system graph $\Sigma$. For instance, the number of qubits $Q(\Sigma)$ is directly determined by the degree of the vertices in $\Sigma$ as
\begin{equation}\label{eq:qubitcount}
    Q(\Sigma)=\sum_{v\in V_\Sigma}n_v=\sum_{v\in V_\Sigma}\ceil{d(v)/2}.
\end{equation}
As is shown in the following, this tight connection between graph-theoretic properties and resource counts holds even for more complicated properties of the fermion-to-qubit encoding and the resulting Hamiltonian simulation.

\section{Parallelization and Path Coloring}\label{s:parallelization}
\subsection{Notions of Parallelization}
In this work, a new possibility for optimization afforded by the flexibility of the custom fermionic codes is considered: parallelization. We use the term parallelization instead of the related concept of circuit depth because our analysis concerns a slightly higher level of abstraction than the particular circuit-level implementation of a Hamiltonian-simulation algorithm. It is assumed that provided 
two Pauli strings $\tilde{h}$ and $\tilde{h}'$ act non-trivially on disjoint sets of qubits, they may be implemented simultaneously in a quantum-simulation algorithm. Therefore, the goal is to minimize the number of steps required to implement the full set of Pauli operators in the interaction set $\tilde{\mathcal{T}}=\{\tilde{h}_\nu\}$. 
If one can choose paths on the system graph for the implementation of the required Pauli strings that minimizes collisions between those paths and orders the implementation of these operators in an optimal way, one can minimize the circuit depth for implementing the relevant operators. See Fig.~\ref{fig:procedure}b for an example. This formulation is especially relevant to quantum simulation via product formulas, in which these Pauli operators are directly implemented for each Trotter step. 

It is important to note that our approach focuses solely on grouping the Pauli strings so as to minimize the number of steps to implement the full interaction set. It is well established
that the choice of ordering terms can impact the Trotter error, which in turn changes the overall circuit depth of the Hamiltonian-simulation algorithm required to achieve a certain error tolerance~\cite{hastings2014improving, tranter2018comparison, Childs2019fasterquantum, tranter2019ordering,nguyen2021digital}. While such effects are not considered in this work, when applying the parallelization techniques here to a particular problem of interest, one should view parallelization of the sort considered here as one piece of a many-faceted optimization.

\subsection{Graph Coloring}
The parallelization problem defined above can be formalized using the notion of path coloring on a graph. This problem also arises in other similar networking and scheduling problems \cite{erlebach2001complexity}. We begin by reviewing the ideas of colorings on graphs and then describe how the parallelization problem may be formulated in these terms. 

Consider a graph $G=\{V, E\}$. A vertex coloring on $G$ is a mapping $\mathcal{C}: V\rightarrow C$ where $C$ is a set of so-called colors or wavelengths. A valid coloring $\mathcal{C}$ is one such that no adjacent vertices in $G$ are assigned the same color. The smallest number of colors required to (vertex) color a graph is called its chromatic number, $\chi(G)$. For a general graph, finding $\chi(G)$ is NP-hard~\cite{karp1972reducibility}. However, both bounds and effective heuristic algorithms exist. A simple and useful upper bound is
\begin{equation}\label{eq:greedybnd}
    \chi(G)\leq \max_{v\in V} d(v)+1,
\end{equation}
where $d(v)$ is the degree of vertex $v$~\cite{brooks1941colouring}. A coloring satisfying this bound can be obtained in polynomial time in the number of vertices using the greedy coloring algorithm presented below.
\begin{algorithm}[H]
 \caption{Greedy Coloring \label{alg:greedycoloring}}
 \begin{algorithmic}[1]
  \Function{GreedyColor}{$G=\{V,E\}, C$}
    \For{each $v \in V$}
        \State Assign $v$ the first color $c\in C$ not used by any of its neighbors
    \EndFor
 \EndFunction
 \end{algorithmic}
\end{algorithm}

If $G$ is a simple, connected graph, but is neither a complete graph nor an odd cycle, then this bound is improved to
\begin{equation}\label{eq:greedybnd2}
    \chi(G)\leq \max_{v\in V} d(v),
\end{equation}
and the greedy coloring algorithm will still satisfy this bound~\cite{brooks1941colouring} . 

The bound on $\chi(G)$ can be still lowered by the clique number $\omega(G)$ of the graph---that is, the size of the largest clique in $G$, where a clique is a complete induced subgraph of $G$. Therefore, the size of any clique $W(G)$ is also a valid lower bound. This gives
\begin{equation}\label{eq:lowerbound}
    \chi(G)\geq \omega(G) \geq |W(G)|.
\end{equation}

A related problem to the vertex-coloring problem is the path-coloring problem. As previously described, this will be our graph-theoretic problem of interest when formalizing the problem of optimally parallelizing the implementation of the Pauli strings that result from a custom fermionic code. In this problem, given a set of paths $\mathcal{P}$ in the graph $G$, one seeks to color the paths such that no two paths which share a vertex in $G$ receive the same color and that a minimum number of colors is used to color all the paths.\footnote{Note that typically in the literature, this problem is defined such that no paths can share an \emph{edge} instead of a vertex. Our alternative definition is due to the particular context in which path coloring is applied.}

The path coloring problem can be mapped to a vertex coloring problem on a different graph called the conflict graph $\Pi(\mathcal{P})$ of the set of paths $\mathcal{P}$. The conflict graph has a vertex set $V_{\Pi(\mathcal{P})}=\mathcal{P}$ and edge set $E_{\Pi(\mathcal{P})}=\{(q,p)\,|\, q,p\in\mathcal{P}, \, q\cap p\neq\emptyset\}$. 
Therefore, the path coloring problem is also NP-hard. 

\subsection{Conflict Graphs for Parallelizability}
Having defined the path coloring, the connection to parallelizability becomes clear. Given a system graph $\Sigma$, one seeks to efficiently implement the interactions in the interaction graph $\Gamma$ as specified by the interaction set $\tilde{\mathcal{T}}$. For any interaction $\tau\in \tilde{\mathcal{T}}$, one requires a choice of path $p$ through $\Sigma$ joining the relevant vertices for the interaction $\tau$. Choosing a particular path for each interaction gives a path set $\mathcal{P}=\{p_\tau\}_{\tau\in \tilde{\mathcal{T}}}$ with $|\mathcal{P}|=|\tilde{\mathcal{T}}|$.
Given a choice of $\mathcal{P}$, one then seeks to determine the degree of parallelization via a coloring of a conflict graph $\Pi(\mathcal{P})$. We construct two different versions of the conflict graph, corresponding to what we dub the \emph{weak coloring problem} and the \emph{strong coloring problem}. The latter considers the internal qubit structure of the vertices of the system graph as specified by the custom fermionic encoding; the former does not. These problems can be formally specified as follows:
\begin{definition}[The weak coloring problem]\label{def:weakcoloring}
Given a system graph $\Sigma$ and a path set $\mathcal{P}$ on $\Sigma$ specifying the implementation of a set of interactions $\tilde{\mathcal{T}}$, construct a conflict graph $\Pi(\mathcal{P})$, whose vertex set is $\mathcal{P}$ and whose edge set is $E_{\Pi(\mathcal{P})}=\{(q,p)\,|\,q,p\in \mathcal{P},\,q\cap p\neq\emptyset\}$. The weak coloring problem is to optimally color $\Pi(\mathcal{P})$. 
\end{definition}
The chromatic number $\chi$ resulting from the weak coloring problem corresponds to the minimum number of steps required to implement all the interactions $\tau\in\tilde{\mathcal{T}}$, where it is assumed that interactions that require disjoint sets of vertices of the system graph may be implemented in parallel.

\begin{definition}[The strong coloring problem]\label{def:strongcoloring}
Given a system graph $\Sigma$ and a path set $\mathcal{P}$ on $\Sigma$ specifying the implementation of a set of interactions $\tilde{\mathcal{T}}$, construct a conflict graph $\Pi(\mathcal{P})$, whose vertex set is $\mathcal{P}$ and whose edge set is $E_{\Pi(\mathcal{P})}=\{(q,p)\,|\,q,p\in \mathcal{P},\,Q(q)\cap Q(p)\neq\emptyset\}$, where $Q(p)$ gives the set of internal qubits required to implement to the path $p\in\mathcal{P}$. The strong coloring problem is to optimally color $\Pi(\mathcal{P})$.
\end{definition}
Note that $Q(p)$ in the definition of the strong coloring problem depends on the local Majorana encoding (i.e., Jordan-Wigner, Fenwick trees, etc.) in the system-graph vertices. This work only considers the Jordan-Wigner encoding of local Majoranas. The following section will provide an explicit description of $Q(p)$ in this setting. Here, the resulting chromatic number $\chi$ corresponds to the minimum number of steps required to implement all the interactions $\tau\in\tilde{\mathcal{T}}$, where it is assumed that interactions that require disjoint sets of qubits may be implemented in parallel. 

Compared to the weak coloring problem, the definition of parallelizability in the strong coloring problem is connected more directly to the qubit architecture and to the circuit depth of the Hamiltonian-simulation algorithm; the weak coloring problem has the advantage of being somewhat more abstracted and easier to work with. Both schemes are considered in this work. Observe also that the definitions of the weak and strong coloring problems take in both the system graph and a particular choice of path for each interaction in the interaction set. This choice of paths, as specified by the set $\mathcal{P}$, is not unique, of course, and to truly maximize the amount of parallelization, one must both pick the optimal path set $\mathcal{P}$ and optimally color the resulting conflict graph. Naturally, this is a very difficult problem. In particular, the following result can be stated:
\begin{fact}
Optimally parallelizing the implementation of an interaction set $\tilde{\mathcal{T}}$---in either the weak or strong coloring sense---is NP-hard.
\end{fact}
\begin{proof}
Suppose there exists an oracle that, given an interaction list and a system graph, returns the solution set of paths $\mathcal{P}$ that will enable the creation of a conflict graph $\Pi(\mathcal{P})$ with the minimum chromatic number. Given $\mathcal{P}$ via this oracle, one is left with a graph coloring problem on $\Pi(\mathcal{P})$, which is known to be NP-hard~\cite{karp1972reducibility}.
\end{proof}
The oracle invoked in the proof above is quite powerful in its own right. Therefore, outside some analytically accessible examples, one need to turn to heuristic algorithms to address the selection of the path set $\mathcal{P}$ and the solving of the resulting weak and strong coloring problems. The full procedure of defining and solving the weak and strong coloring problems starting from the qubit architecture is shown in Fig.~\ref{fig:procedure} for a simple example. Each step of this process will be described in detail in the following sections.
\begin{figure}[t]
    \centering
    \includegraphics[width=0.9\columnwidth]{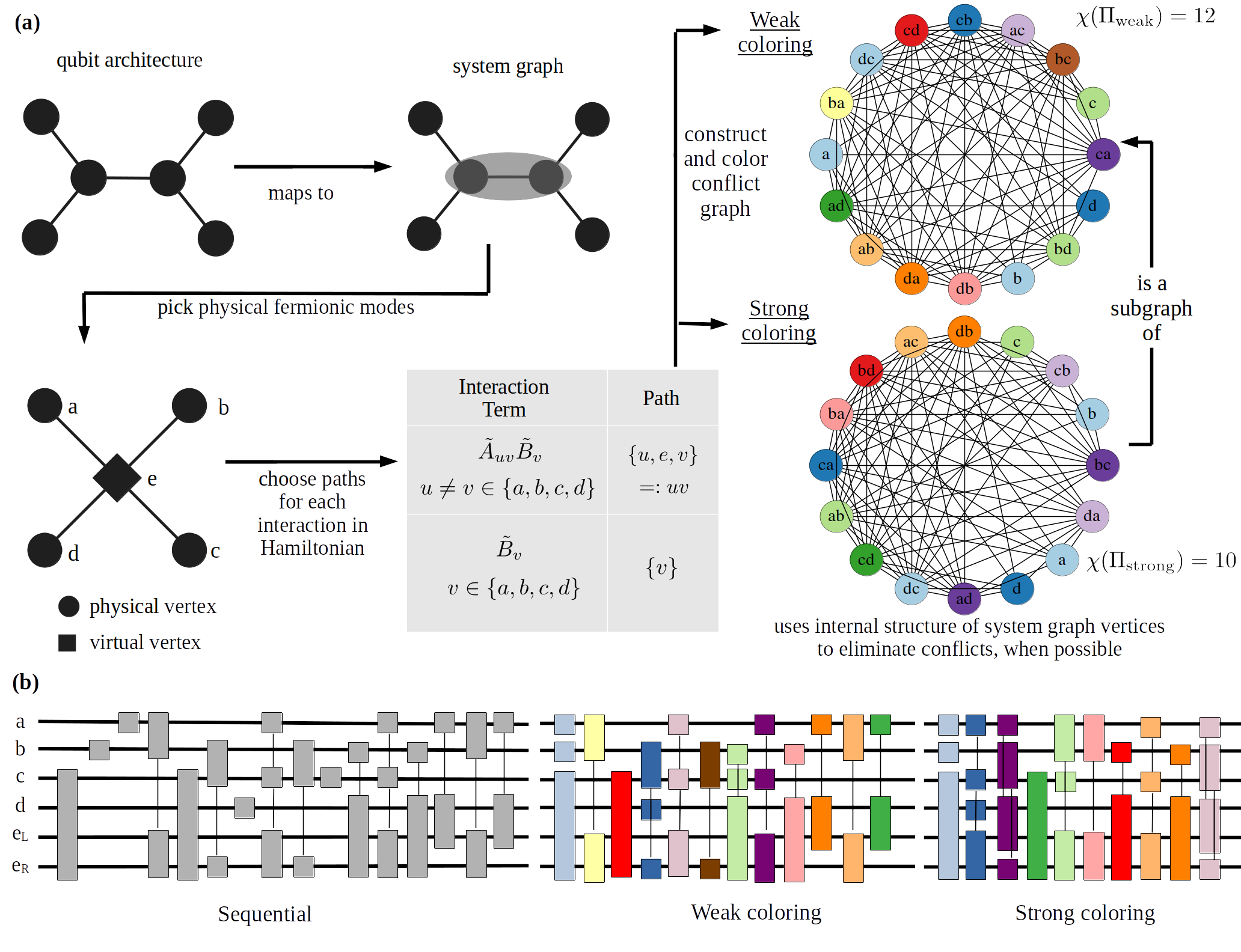}
    \caption{ (a) An overview of the full procedure of defining and solving the weak and strong coloring problems for parallelizing a Hamiltonian simulation of fermions. There are many stages for optimization: the choice of system graph, the choice of physical fermionic modes, the choice of paths linking those modes, and the coloring algorithm. While these choices are straightforward in this small example, for general problems, the design space is extremely large. This work focuses on the last two steps, which is an NP-hard optimization problem. Here, the conflict-graph vertices are labeled by the physical vertices of the system graph involved in the interaction. That is, $\tilde A_{uv}\tilde{B}_v$ is labeled by $uv$ and $\tilde B_v$ is labeled by $v$ for all $u,v$. Note that the difference between the weak and strong coloring problems in this example is in the ability of the strong coloring scheme to route through the virtual vertex $e$ to implement the $ad$ (and $da$) path simultaneously with $bc$ (and $cb$) without any conflict, hence a lower chromatic number compared with the weak coloring scheme. This corresponds to enumerating the edges of vertex $e$ as $\{ea, ed, eb, ec\}\mapsto\{1,2,3,4\}$. (b) Corresponding circuit diagrams for ordering the Pauli strings according to the sequential strategy and via the weak and strong coloring problems. Here $e_L$ and $e_R$ label the left and right internal qubits of vertex $e$ of the system graph, respectively. Colors match those in the corresponding conflict graphs and gates of the same color are implemented simultaneously.}
    \label{fig:procedure}
\end{figure}

\section{Analytic Results}\label{s:analyticresults}
\subsection{The Hamiltonian} 
For the purposes of exploring the weak and strong coloring problems for a variety of system graphs both analytically and numerically, we shall make use of an explicit choice of a fermionic Hamiltonian as a minimal example. In particular, let us consider an all-to-all Hamiltonian with two-mode interactions given by Eq.~(\ref{eq:H}). This Hamiltonian can be expressed in terms of edge and vertex operators as in Eq.~(\ref{eq:alltoall_edge}).
Assuming all coefficients are non-zero, the interaction graph $\Gamma$ for this problem will be the complete graph on $N$ vertices, $K_N$, and the interaction set $\mathcal{T}$ is given by Eq.~(\ref{eq:interactionset}). 
Note that $|\mathcal{T}|=N^2$.

\subsubsection{Extensions to Other Models}
The Hamiltonian in Eq.~(\ref{eq:H}) is closely related to long-range fermionic systems, such as the SYK model~\cite{sachdev1993gapless, rosenhaus2019introduction}. 
To get exact results for specific Hamiltonians of interest (with or without long-range interactions), one can use the algorithm presented here to heuristically solve the weak and strong coloring problems for the relevant system graph. Another example of a minimal fermionic Hamiltonian is that with only nearest-neighbor hopping on a square lattice. This case will be studied later in Sec.~\ref{sec:NN-Hamiltonian}.

A generalization of our results worthy of particular emphasis is the case of Hamiltonians with $k$-body interactions for $k>2$. For instance, terms such as $a_u^\dagger a_v^\dagger a_w a_x$ yield, amongst other things, terms of the form $\tilde{A}_{uv}\tilde{A}_{wx}$ when expressed as edge and vertex operators. Quite clearly, implementing such a term in terms of Pauli operators requires two simultaneous paths through $\Sigma$: one from $u$ to $v$ and one from $w$ to $x$. The path set $\mathcal{P}$ can now be viewed as a multiset of paths, with each element of $\mathcal{P}$ (now potentially a set of paths) mapping to a vertex of the conflict graph. From there, construction of the conflict graph proceeds as usual.

\subsection{Rules for Strong Coloring}
In this section, the rules for constructing the conflict graph for the strong coloring problem given the interactions in Eq.~(\ref{eq:interactionset}) will be developed, under the assumption that local Majoranas are encoded via a Jordan-Wigner transformation on the internal qubits of each vertex of the system graph. This allows to abstract the problem of determining conflicts between paths to one about the properties of the system graph under consideration.

To begin, recall that each vertex $u\in V_\Sigma$ contains $n_u=\ceil{d(u)/2}$ qubits. Under a Jordan-Wigner encoding, one can imagine expanding each vertex of the system graph into a line graph of $n_u$ vertices, where each new vertex is associated with two edges of the original vertex as depicted in Fig.~\ref{fig:catepillar}.\footnote{Note when $d(u)$ is odd, one of these internal vertices has only one external edge.} Any local Majorana operator on the vertex $u$ will induce a Jordan-Wigner string on some subset of these internal vertices. The first task is to identify what precisely these strings are for the four possible (types of) operators acting on the vertex $u$: $\tilde{B}_u$, $\tilde{A}_{xu}$, $\tilde{A}_{xu}\tilde{A}_{uy}$, and $\tilde{A}_{xu}\tilde{B}_u$ where $x,y\in V_\Sigma$ are arbitrary neighbors of $u$ in $\Sigma$. Observe that determining the qubits needed within vertex $u$ to implement the operator $\tilde{A}_{xu}\tilde{B}_u$ is equivalent to $\tilde{A}_{ux}\tilde{B}_u$ since $\tilde{A}_{ux}=-\tilde{A}_{xu}$.
\begin{figure}[h]
    \centering
    \includegraphics[width=0.5\columnwidth]{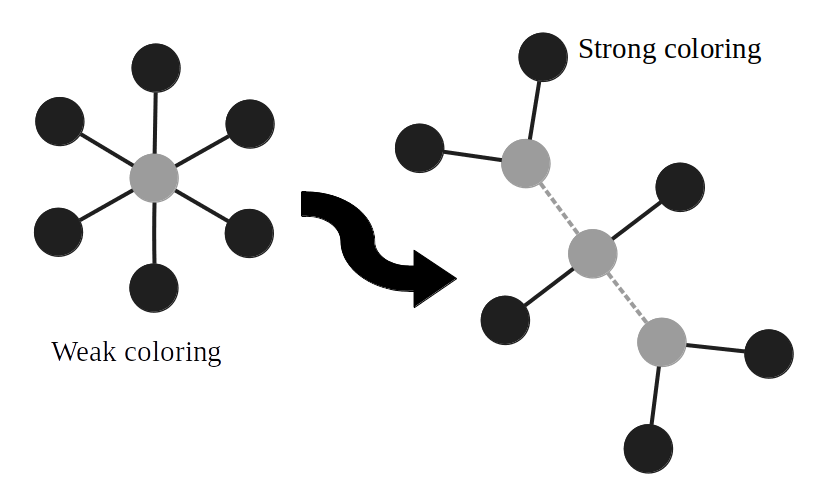}
    \caption{The central gray vertex is $u$ and black vertices are its neighbors. This is just notation and no formal coloring has been performed yet. When constructing the conflict graph for strong coloring under a Jordan-Wigner encoding of the local Majoranas, it is useful to think of each vertex $u\in V_\Sigma$ as being expanded to a line graph of $n_u=\ceil{d(u)/2}$ internal vertices, each connected to two of the original edges of $u$, where $d(u)$ is the degree of $u$. Different interaction types on vertex $u$ induce different Jordan-Wigner strings on these internal vertices as summarized in Tab.~\ref{tab:strongcoloringrules} and depicted in Fig.~\ref{fig:operations}.}
    \label{fig:catepillar}
\end{figure}

First consider a vertex operator $\tilde{B}_u$. Given a Jordan-Wigner encoding of the local Majoranas, one has immediately from Eqs.~(\ref{eq:JWencoding}) and (\ref{eq:qubitvertop}) that in terms of Pauli operators
\begin{equation}
    \tilde{B}_u=\bigotimes_{j=1}^{n_u}Z_u^j,
\end{equation}
where $Z_u^j$ is the Pauli-Z operator acting on qubit $j$ of vertex $u$. Therefore, a vertex operator uses all qubits on that vertex (see Fig.\ref{fig:operations}-a), affording no possibility for improved parallelization via strong coloring when implementing these terms.

On the other hand, operators of the form $\tilde{A}_{ux}$ acting on vertex $u$ do not use all the qubits. Such operators appear when vertex $u$ is a physical vertex and one is seeking to implement an interaction of the form $\tilde A_{uv}\tilde B_{v}$ between vertex $u$ and some other vertex $v$ via a path through $\Sigma$ that starts with the edge from $u$ to $x\in V_\Sigma$. It will be useful to introduce one more piece of notation. In particular, define
\begin{equation}
    a_u(x):=\bigg\lceil\frac{\xi_u(x)}{2}\bigg\rceil,
\end{equation}
so that it can be compactly stated that the first $a_u(x)$ qubits of vertex $u$ are ``active'' when implementing the local Majorana operator ${\tilde{\gamma}}_u^{\xi_u(x)}$. This follows immediately from the Jordan-Wigner encoding of these local Majoranas, where one should recall that the custom fermionic code requires an enumeration of both the internal vertices of $u$ and of its edges. Given a fixed choice of enumeration, $x$ is the $\xi_u(x)$-th neighbor of $u$. Therefore, from Eq.~(\ref{eq:qubitedgeop}), one immediately finds that the operator $\tilde{A}_{ux}$ makes use of the first $a_u(x)$ qubits of vertex $u$ (as well as the first $a_x(u)$ qubits of vertex $x$), see Fig.~\ref{fig:operations}-b.

Next consider an operator of the form $\tilde{A}_{xu}\tilde{A}_{uy}$ acting on vertex $u$. Such operators occur when vertex $u$ is an intermediate vertex along a path implementing an interaction between two physical fermionic modes. Just like $\tilde{A}_{ux}$, these operators also do not require the use of all qubits in $u$. Individually, $\tilde{A}_{xu}$ and $\tilde{A}_{uy}$ make use of the first $a_u(x)$ and the first $a_u(y)$ qubits in $u$, respectively. However, there are cancellations since the operators both act with Pauli-Z operators on the first $a_u(x)-1$ and $a_u(y)-1$ qubits, respectively. Such cancellations of Jordan-Wigner strings are reminiscent of the cancellations of such strings in sequential Trotter-Suzuki steps~\cite{hastings2014improving}. The net result is that only the qubits between $\min\{a_u(x), a_u(y)\}$ and $\max\{a_u(x), a_u(y)\}$ are used, see Fig.~\ref{fig:operations}-c.

Finally, consider an operator of the form $\tilde{A}_{xu}\tilde{B}_u$ acting on vertex $u$. These operators arise at the starting and ending vertices of a path. Once again, there are cancellations in the Pauli-Z operators required to implement the two sub-operators. In particular, $\tilde{A}_{xu}\tilde{B}_u$ acts on the last $n_u-a_u(x)+1$ qubits of vertex $u$, see Fig.~\ref{fig:operations}-d.
Tab.~\ref{tab:strongcoloringrules} summarizes the results in this section.
\begin{figure}
    \centering
    \includegraphics[width=\columnwidth]{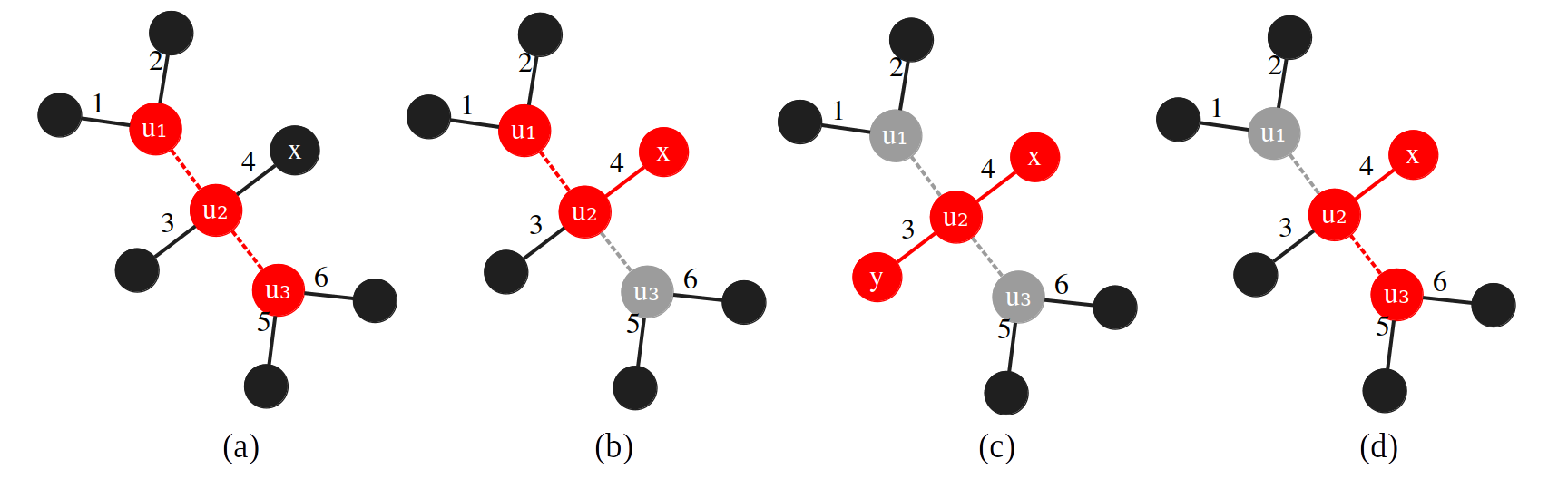}
    \caption{Examples of the internal qubits of vertex $u$ (enumerated top to bottom as $u_1, u_2, u_3$) that are required to implement the operators (a) $\tilde{B}_{u}$, (b) $\tilde{A}_{ux}$, (c) $\tilde{A}_{xu}\tilde{A}_{uy}$, and (d) $\tilde{A}_{xu}\tilde{B}_{u}$. Active qubits and input/output system-graph vertices are marked with red. The dashed lines denote internal edges. Again, the coloring is the notation and no formal graph coloring is assumed here. The choice of edge enumeration is marked. }.
    \label{fig:operations}
\end{figure}
\begin{table}[]
    \centering
    \begin{tabular}{|c|c|}
        \hline
        Term &  Active qubits in vertex $u$\\
        \hline
        $\tilde{B}_{u}$& All $n_u$ qubits \\
        $\tilde{A}_{ux}$& Qubits 1 to $a_u(x)$ \\
        $\tilde{A}_{xu}\tilde{A}_{uy}$& Qubits $\min\{a_u(x), a_u(y)\}$ to $\max\{a_u(x), a_u(y)\}$ \\ 
        %& and $\max\{a_u(x), a_u(y)\}$ \\
        $\tilde{A}_{xu}\tilde{B}_{u}$& Qubits $a_u(x)$ to  $n_u$ \\
        \hline
    \end{tabular}
    \caption{Rules for determining internal qubits used by the various terms that arise in simulating the Hamiltonian in Eq.~(\ref{eq:alltoall_edge}) using a Jordan-Wigner encoding of the local Majoranas. Recall $a_u(v):=\ceil{\frac{\xi_u(v)}{2}}$ is the number of ``active'' qubits when implementing ${\tilde{\gamma}}_u^{\xi_u(v)}$, and $v$ is the $\xi_u(v)$-th neighbor of $u$.}
    \label{tab:strongcoloringrules}
\end{table}

\subsection{Limits of Weak and Strong Coloring}\label{ss:limits}
In this section, two simple system graphs will be studied: a star graph $S_N$ with $N$ physical vertices all joined to a central virtual vertex, and a complete graph $K_N$ consisting of $N$ physical vertices. These examples are limiting cases for both the weak and strong coloring problems. In addition, they allow for straightforward analytic calculations and enable an understanding of the essential conceptual features of the two types of coloring problems. This understanding will be leveraged to determine what properties of a system graph allow for the greatest possible improvement from using strong coloring as opposed to weak coloring. An example with such an extreme separation will be constructed at the end of the section.

\subsubsection{Star Graph}
The star graph $\Sigma=S_N$ of $N$ physical vertices all coupled to a central virtual vertex is the worst-case limit for parallelization as there is a single bottleneck vertex through which all paths for the $N(N-1)$ two-mode interactions must pass. It helps to refer back to Fig.~\ref{fig:procedure} to visualize the procedure for the minimal case of $\Sigma=S_4$. As seen in that figure, the corresponding conflict graph for the weak coloring problem has a complete subgraph $K_{N(N-1)}$ consisting of all vertices that correspond to two-mode interactions, which sets a lower bound on the chromatic number of the conflict graph. No additional colors are needed to color the one-mode interaction vertices as the vertex operators $\{\tilde{B}_w\}_{w\notin\{u,v\}}$ can be implemented simultaneously with the $\tilde{A}_{uv}\tilde{B}_v$ operators. Therefore, the chromatic number for weak coloring is 
\begin{equation}\label{eq:weakstar}
    \chi_\mathrm{weak}(\Pi(S_N))=N(N-1).
\end{equation}

For strong coloring, it turns out that the even and odd $N$ cases must be addressed separately. First, consider $N$ even. Expanding all vertices of the system graph as in Fig.~\ref{fig:catepillar}, the physical vertices remain unexpanded, whereas the central virtual vertex $u$ expands to a line graph of $n_u=N/2$ vertices $\{u_1, \cdots u_{n_u}\}$, where each $u_j\in u$ has two edges that each connect to one neighbor of $u$ in $\Sigma$. The two-mode interactions involving vertex $u$
%each
 induce eight Jordan-Wigner strings between each pair of these neighbors---two for each of the four choices of pairs of physical-neighbor vertices connected to a given pair $(u_\mu, u_\nu)$, $\mu\neq \nu$. They also induce single-vertex ``strings'' for each $u_\nu$ for the two-mode interactions between physical vertices that are both neighbors of that vertex. 

Minimizing the number of steps to avoid overlaps of these strings is straightforward: starting with $u_1$, implement all interactions that induce Jordan-Wigner strings originating from $u_1$ while simultaneously implementing the interaction that induces the longest non-overlapping Jordan-Wigner string originating from $u_{n_u}$. This takes $8n_u=4N$ steps. At this point, all interactions involving $u_1$ and $u_{n_u}$ have been implemented. Therefore, ignore those vertices and repeat the same procedure on the remaining $n_u-2$ internal vertices. Keep repeating this procedure until all two-mode interactions have been implemented. For $N>2$, implementing the single-mode interactions requires no extra steps as most physical vertices are unused for any given step, giving many opportunities to implement these interactions simultaneously with a given two-mode interaction. The net result (for $N>2$, even) is
\begin{equation}\label{eq:strongstarNeven}
    \chi_\mathrm{strong}(\Pi(S_N), N\,\mathrm{even})=\begin{cases}
    8 \sum_{\mu=1}^{n_u/2} (2\mu)-6 = \frac{N^2}{2}+2N -6, & \frac{N}{2}\, \text{even}, \\
    \\
    8 \sum_{\mu=1}^{(n_u+1)/2} (2\mu-1)-6 = \frac{N^2}{2}+2N-4, & \frac{N}{2}\, \text{odd}. \\
    \end{cases}
\end{equation}
Note that number six is subtracted from the sum to correct for over-counting in the final step which only involves two-mode interactions between physical vertices that share an internal vertex. See Fig.~\ref{fig:strongcoloringstar} for an example of this construction for $N=8$.

Now consider $N$ odd. Expanding the vertices of the system graph, the result is identical to the even-$N$ case except $u_{n_u}$ has only one edge joining it to a physical vertex. This implies that the two-mode interactions involving $u_{n_u}$ induce only four Jordan-Wigner strings between the internal vertices of $u$ instead of four. One can implement these interactions first while simultaneously implementing four of the eight interactions that induce the longest possible non-overlapping Jordan-Wigner strings starting from $u_1$. This takes $4n_u=2(N+1)$ steps. At this point, all interactions that involve $u_{n_u}$ are implemented, but four interactions are yet to be implemented for each induced Jordan-Wigner string involving $u_1$. These can be implemneted while simultaneously implementing four of the eight interactions that induce the longest possible non-overlapping Jordan-Wigner strings starting from $u_{n_u-1}$. This takes $4(n_u-1)$ steps. This staggered approach can be continued---implementing four of the eight interactions that induce a particular Jordan-Wigner string starting from a given internal vertex in each stage of the procedure---until all interactions are implemented. This gives
\begin{equation}\label{eq:strongstarNodd}
    \chi_\mathrm{strong}(\Pi(S_N), N\,\mathrm{odd})=4\sum_{\mu=1}^{n_u-1} (\mu+1)-2=
    \frac{N^2}{2}+2N-\frac{9}{2},
\end{equation}
where number two is subtracted from the sum to correct for over-counting in the final step of this procedure. See Fig.~\ref{fig:strongcoloringstar} for an example of this construction for $N=7$. 

The constructions yielding Eqs.~(\ref{eq:strongstarNeven}) and (\ref{eq:strongstarNodd}) are optimal. In particular, observe that in each step of these constructions, a path is implemented that passes through the central internal vertex $u_{\ceil{n_u/2}}$ of the central vertex $u$. The set of paths that go through this vertex form a complete induced subgraph of the conflict graph $\Pi_\mathrm{strong}(S_N)$ of maximum size. That is, the induced subgraph of this set of vertices in the conflict graph form the largest-size clique. Since one of these paths is implemented in every step of the construction, the corresponding coloring of $\Pi_\mathrm{strong}(S_N)$ saturates the clique-number lower bound on the chromatic number from Eq.~(\ref{eq:lowerbound}). Therefore, the constructions are optimal.
\begin{figure}
    \centering
    \includegraphics[width=\columnwidth]{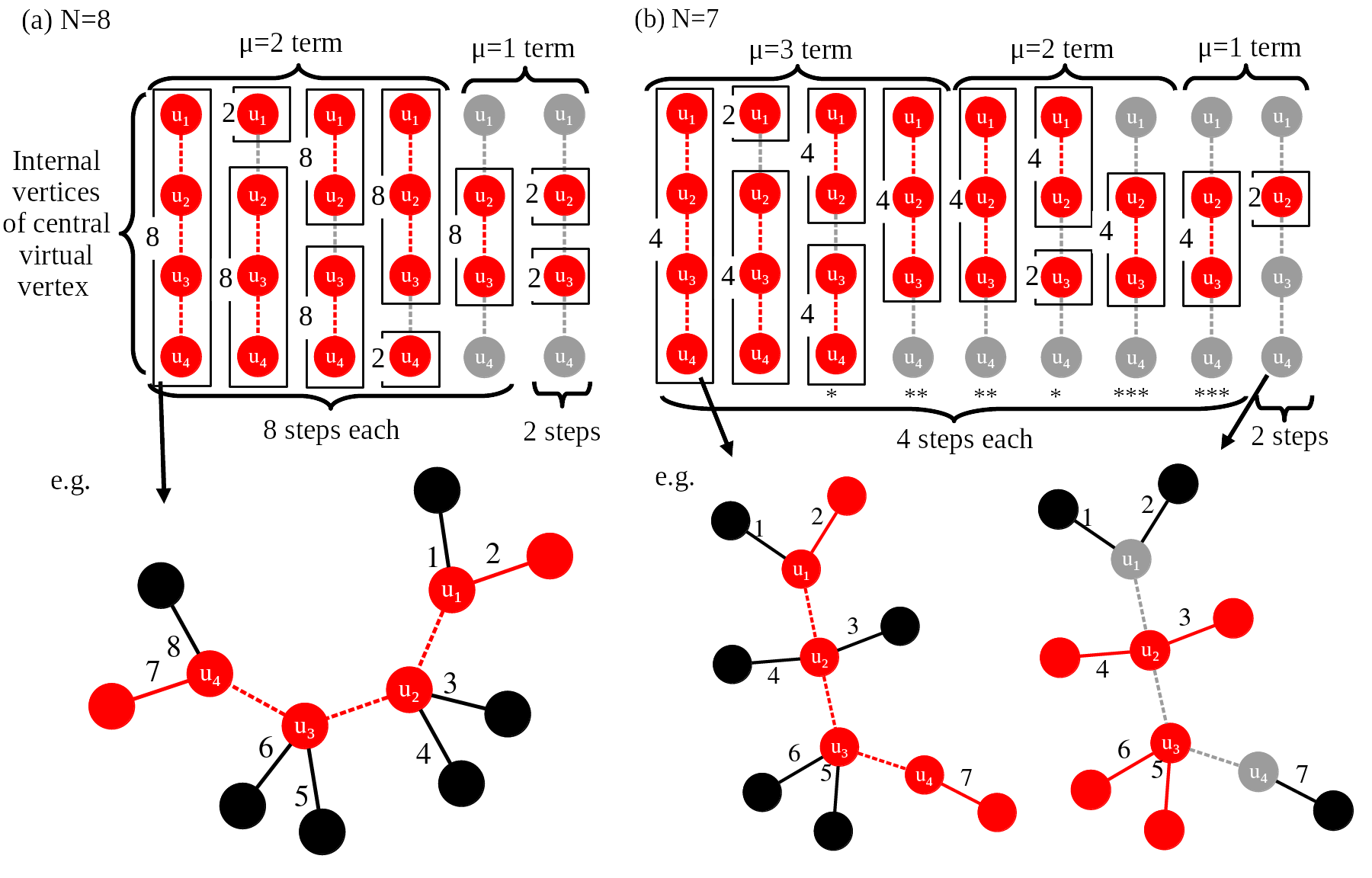}
    \caption{Examples of the optimal procedure to parallelize two-mode interactions via strong coloring for the star graph for (a) $N=8$ and (b) $N=7$. In the top of the figure, the internal vertices of the central virtual vertex $u$ are shown. Edges between internal vertices are represented by dashed lines. Red (gray) vertices are active (inactive) with red (gray) lines indicating induced (no) Jordan-Wigner strings. Terms are grouped as they appear in the respective sums over $\mu$ in Eqs.~(\ref{eq:strongstarNeven})-(\ref{eq:strongstarNodd}). Observe in the case of $N$ odd, this means that the Jordan-Wigner strings of a given type are split in the different groupings (the locations of such splits are marked by one, two, or three stars). The step counts underneath each group give the number of steps to implement all interactions that induce that set of Jordan-Wigner strings. This gives a total of 42 and 34 steps for $N=8$ and $N=7$, respectively. The counts associated with each individual Jordan-Wigner string give the number of interactions corresponding to that Jordan-Wigner string. This gives a total of $8\times(8-1)=56$ and $7\times(7-1)=42$ interactions for $N=8$ and $N=7$, respectively.  Representative examples of the types of interactions between physical vertices that induce the different Jordan-Wigner strings are shown in the bottom of the figure. Observe in (b) that for $N$ odd, the last internal vertex has only one physical vertex as a neighbor. This is responsible for the different procedure for optimal parallelization.}
    \label{fig:strongcoloringstar}
\end{figure}

\subsubsection{Complete Graph}\label{ss:complete}
The complete graph $\Sigma=K_N$ of $N$ all-to-all connected physical vertices is the opposite limit of the star graph. There are no bottlenecks in implementing paths between any pair of vertices as all two-body interactions are directly implementable. Obviously, these direct connections are the optimal paths. 

For the weak coloring problem, one can simultaneously implement $\floor{N/2}$ of the $N(N-1)$ two-mode interactions. Consequently, all two-mode interactions can be implemented in $N(N-1)/\floor{N/2}$ steps. In addition, there are $N$ one-mode terms to implement. For even $N$, these can all be implemented in one step after doing the two-mode interactions. For odd $N$, there is always an unused vertex for any step where two-mode interactions are implemented, and therefore the one-mode interactions can be done while doing the two-mode interactions. Therefore, the number of steps required for each case is
\begin{align}\label{eq:completeweak}
    \chi_\mathrm{weak}(\Pi(K_N))=\begin{cases}
    2N-1, & N\text{ even},\\
    2N, & N\text{ odd}.
    \end{cases}
\end{align}

To gain some intuition about the conflict graph, we can also arrive at Eq.~(\ref{eq:completeweak}) by observing that the conflict graph $\Pi_{\mathrm{weak}}(K_N)$ consists of $N$ pairwise-overlapping complete subgraphs $K_{2N-1}$, as depicted in Fig.~\ref{fig:complete_weak} for the case of $N=4$. This structure arises because for any $v\in V_{\Sigma}$, there are $2N-1$ interactions involving this vertex which all mutually conflict. Pairwise overlaps occur between these complete subgraphs because each vertex in $\Pi_\mathrm{weak}(K_N)$ corresponding to a two-mode interaction between $u,v\in V_\Sigma$ is in the complete subgraph corresponding to both $u$ and $v$. 

Let each of the $N$ complete subgraphs be the vertices of a new (complete) graph $G$ with two edges between each $v\in G$.  One can then map the problem of coloring the two-mode-interaction vertices to one of edge coloring $G$ so that no edges sharing a vertex share a color. In particular, one may color the two-mode interaction vertices of $\Pi_\mathrm{weak}(K_N)$ with the color of the corresponding edge of $G$. Edge coloring $K_{N}$ takes $N-1$  colors for $N$ even and $N$ colors for $N$ odd~\cite{soifer2009mathematical}. Due to the double edges between each $(u,v)\in V_G$, twice this number is required. Finally, one must consider coloring the one-mode interaction vertices in the original problem. For even $N$, this requires an additional color because all $2(N-1)$ colors are used in each complete subgraph. For odd $N$, each complete subgraph has two unused colors, and therefore, one can use one of these colors for the single-mode vertex. This recovers Eq.~(\ref{eq:completeweak}) for the number of steps required for weak coloring.
\begin{figure}
    \centering
    \includegraphics[width=0.705\columnwidth]{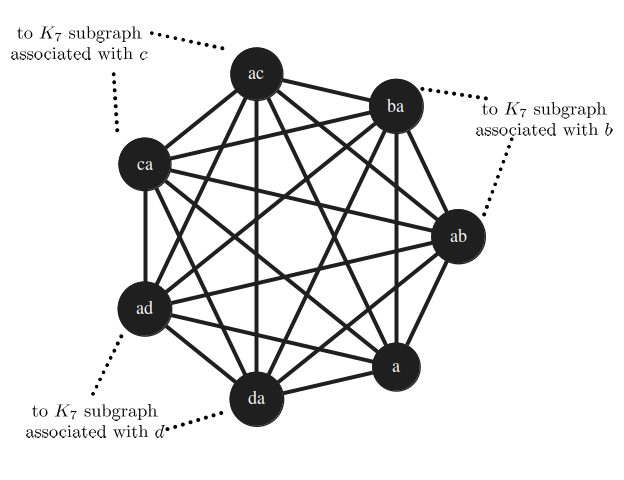}
    \caption{An illustration of the conflict graph for the weak-coloring problem with a system graph $K_N$ with $N=4$. The vertices of the system graph are labeled as $\{a,b,c,d\}$ (e.g., as in Fig.~\ref{fig:hamiltoniancycle}), and the vertices of the resulting conflict graph are labeled by the corresponding interaction. The conflict graph consists of $N=4$ interlocking complete subgraphs each associated with one of the system-graph vertices, as described in the main text. One such complete subgraph is shown.}
    \label{fig:complete_weak}
\end{figure}

Consider strong coloring for this problem. As all interactions can be directly implemented along a single edge of the system graph, one is limited by the capacity of the physical vertices to have multiple inputs and outputs. In particular, referring to Tab.~\ref{tab:strongcoloringrules}, it is clear that multiple ``ingoing'' ($\tilde A_{xu}\tilde B_u$) or ``outgoing'' ($\tilde A_{ux}$) interactions cannot be implemented simultaneously on a given vertex $u$. 
However, one can simultaneously have one ``ingoing'' and one ``outgoing'' interaction for a given vertex---that is, a term of the form $\tilde A_{ux}$ and a term of the form $\tilde A_{yu}\tilde B_{u}$ can be simultaneously implemented on vertex $u$ provided that $a_u(x)<a_u(y)$. 

A lower bound on the chromatic number of $\Pi_\mathrm{strong}(K_N)$ can be obtained in terms of the clique number of the graph (see Eq.~(\ref{eq:lowerbound})). In particular, one finds that
\begin{equation}\label{eq:lowerboundcomplete}
    \chi_\mathrm{strong}(\Pi(K_N))\geq N+2.
\end{equation}
This bound is derived as follows: Given any vertex $u\in V_{\Sigma}$, the set of all $N-1$ interactions $\tilde A_{uv}\tilde B_v$ for all $u\neq v$, the interaction $\tilde B_u$, and two of the $\tilde A_{vu}\tilde B_u$ interactions all require the use of the first internal qubit $u_1$. These interactions form the largest complete subgraph $K_{N+2}$ of $\Pi_\mathrm{strong}(K_N)$. Coloring this complete subgraph requires $N+2$ colors, yielding Eq.~(\ref{eq:lowerboundcomplete}). Therefore at best, asymptotically (in $N$) one obtains $\chi_\mathrm{weak}(\Pi(K_N))/\chi_\mathrm{strong}(\Pi(K_N))\sim 2$.

An upper bound on the chromatic number of $\Pi_\mathrm{strong}(K_N)$ can be found by explicit construction. For any Hamiltonian cycle\footnote{A Hamiltonian cycle on a graph is a cycle (closed loop) through the graph that visits each vertex exactly once.} on $K_N$, the edges in the cycle can be enumerated in such a way that all $N$ interactions along these edges can be implemented simultaneously via strong coloring. See Fig.~\ref{fig:hamiltoniancycle} for an example of this for $K_4$. The number of edge-disjoint Hamiltonian cycles on a complete graph is $(N-1)/2$ for odd $N$ and $(N-2)/2$ for even $N$~\cite{bryant2007cycle, alspach2008wonderful}. One can independently enumerate these edge-disjoint Hamiltonian cycles such that all interactions within each of these cycles can be implemented simultaneously.  Each of these disjoint Hamiltonian cycles are then sequentially implemented. Assuming no other improvements from strong coloring over weak coloring gives an obtainable upper bound, as described below.
\begin{figure}
    \centering
    \includegraphics[width=0.75\columnwidth]{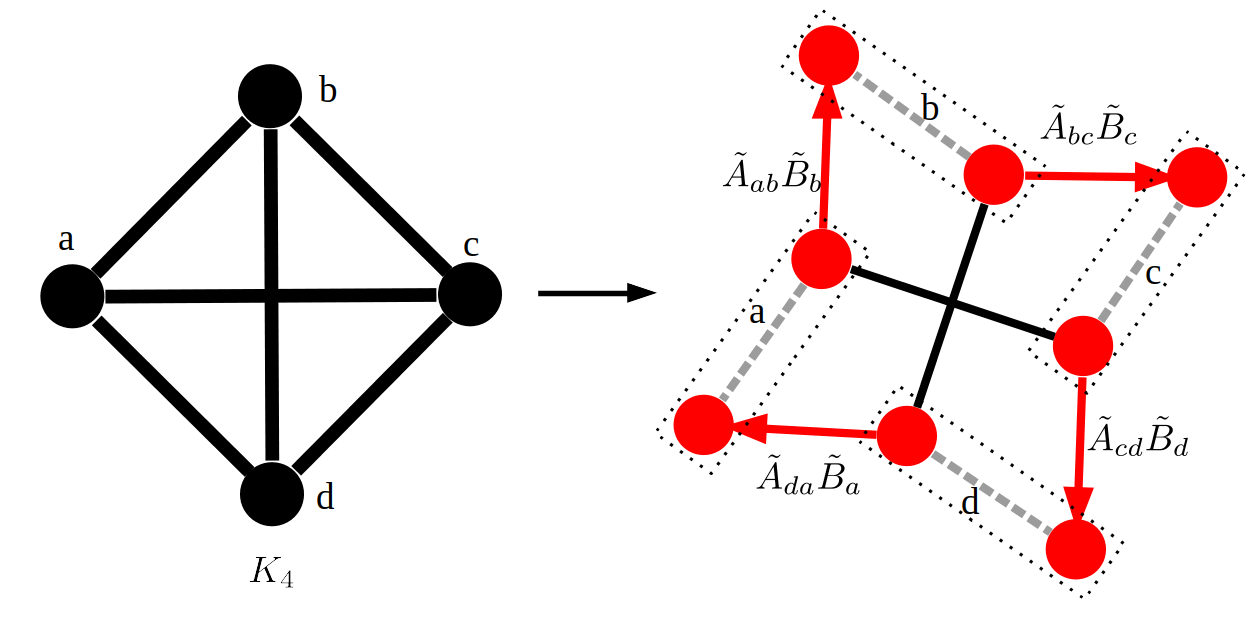}
    \caption{A minimal example of $K_4$ that shows a complete system graph can have its edges enumerated such that a Hamiltonian cycle of two-mode interactions can be implemented simultaneously via strong coloring. The dashed lines denote internal edges, red denotes active qubits, and red arrows from $u$ to $v$ for $u,v\in V_\Sigma$ denote the implementation of an interaction of the type $\tilde{A}_{uv}\tilde B_{v}$.}
    \label{fig:hamiltoniancycle}
\end{figure}

Let us consider odd $N$ first. Once the interactions contained in all $(N-1)/2$ edge-disjoint Hamiltonian cycles of $K_N$ are implemented, exactly half of the two-mode interactions are completed in $(N-1)/2$ steps and each edge is traversed exactly once. Considering the rest of the interactions in terms of weak coloring, the problem can be reduced to edge coloring a complete graph as described above. The only difference is that one no longer has parallel edges to consider since one of the two interactions along every edge $e\in E_\Sigma$ is already implemented. This gives
\begin{equation}
    \chi_\mathrm{strong}(K_N, N\,\mathrm{odd})\leq \frac{3N-1}{2}.
\end{equation}

For even $N$, once the interactions contained in all $(N-2)/2$ edge-disjoint Hamiltonian cycles of $K_N$ are implemented, one is still left with some parallel edges in the edge-coloring formulation of the weak coloring problem. The extra parallel edges form a perfect matching\footnote{A perfect matching is a set of pairwise non-adjacent edges that cover every vertex of the graph}~\cite{bryant2007cycle, alspach2008wonderful} and therefore these ``extra'' interactions can be colored with a single additional color. The problem now reduces to the no-parallel edges version of the edge coloring problem on a complete graph, yielding a final upper bound of
\begin{equation}\label{eq:completestrongupper}
    \chi_\mathrm{strong}(K_N, N\,\mathrm{odd})\leq \frac{3N}{2}.
\end{equation}

Given this explicit construction, the combined asymptotic bounds on the improvement from strong coloring over weak coloring is
\begin{equation}
  \frac{3}{2} \lesssim \frac{\chi_\mathrm{weak}(K_N)}{\chi_\mathrm{strong}(K_N)}\lesssim 2.
\end{equation}

\subsubsection{Separating Weak and Strong Coloring} \label{s:separation}
For the star-graph and the complete-graph examples, it was observed that asymptotically (in $N$) $\frac{\chi_\mathrm{weak}}{\chi_\mathrm{strong}}\lesssim 2$. 

As it will be seen numerically in Sec.~\ref{s:numericalresults}, such constant-factor improvements are typical for system graphs that arise from realistic qubit architectures. In the near term, eliminating such constant overheads in circuit depth is important and serves as one of the practical motivations for this work, but a more significant separation in parallelization performance between weak and strong coloring can be demonstrated. This example, while contrived, serves to show that polynomial separations are possible and, perhaps more importantly, highlights a key feature of system graphs which allow for a large separation between weak and strong coloring. 

In particular, a necessary condition for a large improvement due to strong coloring is that the edge bottleneck(s) of the system graph for routing paths in $\tilde{\mathcal{T}}$ are significantly larger than the vertex bottleneck(s). The reason for this is clear: vertex bottlenecks are the limiting factor on parallelization for the weak coloring problem, whereas edge bottlenecks are the limiting factor for strong coloring. That is, the more vertices (edges) to route paths through in the weak (strong) coloring problems, the more room there is for parallelization. When there is a large separation between the size of the edge and vertex bottlenecks, strong coloring necessarily provides more of an advantage. The star graph is a simple example of such a large separation between edge and vertex bottlenecks. The vertex bottleneck is a single vertex, but many edges enter this vertex, suggesting a large potential improvement via strong coloring. We know analytically that this improvement is asymptotically a factor of two. Despite the large separation between edge and vertex bottlenecks, most interactions passing through the central vertex still require many of the internal qubits, hence limiting a greater potential for strong coloring. The counterexample constructed below aims to avoid this limitation.

For simplicity, let us consider $N=4m$ for some positive integer $m$. Divide the physical vertices of $\Sigma$ corresponding to these $N$ fermionic modes into two disjoint sets $T_1, T_2$, each of size $N/2$. Consider adding edges to the system graph such that both $T_1$ and $T_2$ induce disjoint complete subgraphs. 
Now, consider adding $N/2$ virtual vertices to each of these complete subgraphs. Add an edge from each of these virtual vertices to all physical vertices in the subgraph, forming two bipartite subgraphs. Finally, add a single additional virtual vertex and join it to all other virtual vertices. See Fig.~\ref{fig:polyexample} for an example of the construction for $N=8$.

\begin{figure}
    \centering
    \includegraphics[width=0.725\columnwidth]{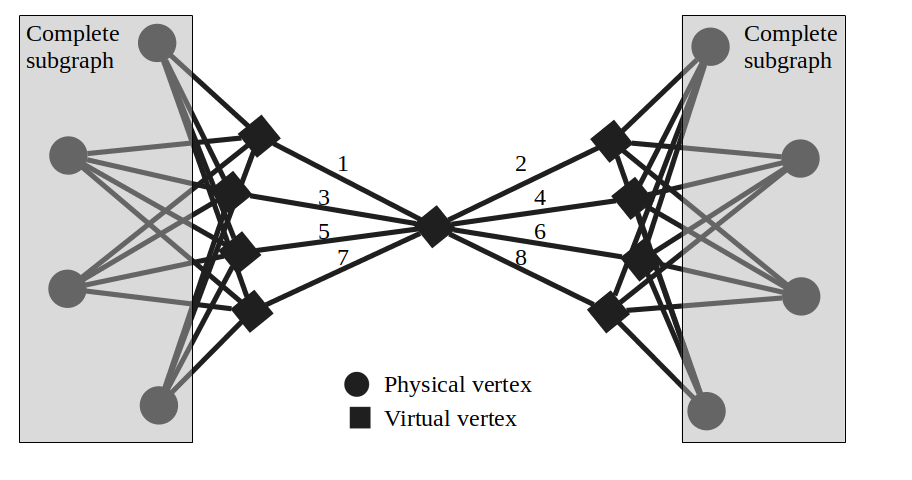}
    \caption{An example of the system graph $\Sigma_\text{bottleneck}$ that allows for a linear-in-$N$ scaling of the ratio of the chromatic numbers of the conflict graphs with weak and strong coloring. 
    The physical vertices are divided into two complete subgraphs $K_{N/2}$ separated by a single vertex bottleneck, which allows only one interaction at a time between the two subgraphs in the weak coloring problem. For visual clarity, the edges between vertices of the complete subgraphs are not shown. For strong coloring, the extra layer of virtual vertices between each subgraph and the central bottleneck vertex allows any disjoint set of $\frac{N}{2}$ interactions between the two subgraphs to be implemented in a single step.}.
    \label{fig:polyexample}
\end{figure}

Call this system graph $\Sigma_\text{bottleneck}$. $\Sigma_\text{bottleneck}$ has a single vertex bottleneck between its two symmetric halves. The weak coloring chromatic number can computed to be
\begin{equation}
    \chi_\mathrm{weak}(\Pi(\Sigma_\text{bottleneck}))= \underbrace{2\left(\frac{N}{2}\right)^2}_{\textrm{\shortstack{paths between\\ $T_1, T_2$}}}+\underbrace{N-1.}_{\textrm{\shortstack{paths within \\ $T_1, T_2$}}}
\end{equation}
Now consider enumerating the edges of the central virtual vertex so that all edges going to one half of the graph are even integers and all edges going to the other half are odd integers. With this labeling, for any choice of $N/2$ interactions from one half of the graph to the other, one can route all $N/2$ interactions through the central vertex simultaneously using strong coloring. Applying only weak coloring to implement interactions within each complete subgraphs then gives
\begin{equation}
    \chi_\mathrm{strong}(\Pi(\Sigma_\text{bottleneck}))\leq \underbrace{N}_{\textrm{\shortstack{paths between\\ $T_1, T_2$}}}+\underbrace{N-1}_{\textrm{\shortstack{paths within \\ $T_1, T_2$}}}=2N-1,
\end{equation}
which yields $\frac{\chi_\mathrm{weak}(\Pi(\Sigma_\text{bottleneck}))}{\chi_\mathrm{strong}(\Pi(\Sigma_\text{bottleneck}))}\gtrsim \frac{N}{4}$. Such linear-in-$N$ improvement from strong coloring is the best possible scaling for this ratio as the separation between sequential implementation of all interaction terms and the best possible parallel scheme is $\sim N$.

One should be cautious in interpreting this large separation. In practical settings, intelligent design of system graphs from the underlying qubit architecture will rule out such large separations. In practice, more modest, but important, constant-factor improvements between strong and weak coloring should be expected. In particular, there is no reason why the qubits in the central virtual vertex of this example should all be grouped into one system-graph vertex---there are no interactions that require more than a single qubit operator within this vertex. Therefore, a more intelligent system graph built on the same underlying qubit structure would afford the weak coloring problem access to the same performance as the strong coloring problem in our contrived example, by splitting the central virtual vertex into $N/2$ vertices. 

Consequently, this example also raises the issue of intelligent system-graph design as a prerequisite to using our algorithms to greatest effect. In Sec.~\ref{s:numericalresults}, we give some more examples of how system graphs may be constructed from the underlying qubit architectures.

\section{Numerical Results}\label{s:numericalresults}
\subsection{Description of the Algorithm}
For more complicated examples, it is necessary to turn to heuristic algorithms to find good solutions to the parallelization problem in either the weak or strong coloring schemes. Any such algorithm must perform the following steps: First, it must identify a set $\mathcal{P}$ of paths between interacting vertices. Then, given $\mathcal{P}$, it must construct the corresponding conflict graph, which---in the case of strong coloring---requires a choice of enumerating the edges of the system graph. It is known from our analytic results that this choice of enumeration can have a significant impact. Finally, the algorithm must perform a vertex coloring on the resulting conflict graph, which is well-known to be an NP-hard problem in its own right. 

Our algorithmic approach to these problems is largely a straightforward one. The most essential and novel aspect of the algorithm relates to choosing the enumeration of edges in an intelligent way to amplify the improvement from strong coloring over weak coloring as much as possible. This is important because it is this step that allows one to parallelize simulation of fermionic Hamiltonians in a way that is aware of the fermion-to-qubit mapping chosen. An understanding of the advantages afforded by taking this information into account is one of the primary goals of this paper. An overview of the salient features of the algorithm is provided here, and the reader is referred to the $\mathtt{github}$ repository for access to the full code~\cite{github}.

In either the weak or strong coloring case, to find a path set $\mathcal{P}$, one can begin by weighting the edges of $\Sigma$ to penalize edges that connect to physical vertices since physical vertices are used in physical interactions and may need to be saved for the implementation of other terms that involve them. The exact amount of this penalty is a free parameter of the algorithm. The larger the penalty, the more the algorithm will prioritize potential for parallelization over minimizing the Pauli weight of operators. This is because the use of physical vertices may provide shorter paths and, hence, shorter Pauli strings but those are penalized by the algorithm.

Next the algorithm needs to choose a random ordering of the interaction set $\tilde{\mathcal{T}}$. Given this ordering, for each interaction $\tau\in\tilde{\mathcal{T}}$, the algorithm identifies a path as the shortest distance, weighted path through $\Sigma$ connecting the relevant vertices. This can be done efficiently in time $\Theta\big((|V_\Sigma|+|E_\Sigma|)\log |V_\Sigma|\big)$ via Dijkstra's algorithm \cite{mehlhorn2008algorithms}. Next, the weight of all edges used in this path are increased and the algorithm proceeds to finding the shortest path for the next $\tau\in\tilde{\mathcal{T}}$ on the reweighted graph. The increase in the weight of the used vertices penalizes paths that do not find ``new'' routes through $\Sigma$---this is advantageous since paths that traverse the same edge in $\Sigma$ are guaranteed to conflict. The exact choice for this penalty is, again, a free parameter of the algorithm. 
Due to the sequential nature of this algorithm and the penalties for traversing previously used edges, different orderings of $\tilde{\mathcal{T}}$ will produce different path sets $\mathcal{P}$. Consequently, one needs to run the algorithm many times to sample a variety of different path sets.

Once $\mathcal{P}$ is generated, the algorithm go on to construct the corresponding conflict graph $\Pi(\mathcal{P})$. For weak coloring, this is straightforward---if two paths $p,q\in \mathcal{P}$ share any vertices, the corresponding vertices in $\Pi(\mathcal{P})$ share an edge. For strong coloring, whether or not $p,q\in \mathcal{P}$ conflict depends on the choice of edge enumeration for the vertices in the paths. This choice of enumeration is arbitrary, so a wise choice is an enumeration that attempts to minimize conflicts. In particular, whenever the algorithms finds a path $p\in\mathcal{P}$, it loops through the vertices $u\in p$ and enumerates any previously unenumerated edges according to the following rules: If $v$ is the first vertex in the path and therefore acted on by an operator of type $\tilde{A}_{ux}$, it enumerates the outgoing edge with the smallest available index. If $u$ is an interior vertex along the path and therefore acted on by an operator of type $\tilde{A}_{xu}\tilde{A}_{uy}$, the algorithm enumerates the incoming vertex as the smallest available index and the outgoing edge as the next smallest available index. Finally, if $u$ is the final vertex in the path and therefore acted on by an operator of type $\tilde A_{xu}\tilde B_u$, the algorithm enumerates the incoming edge with the largest available index. This method of constructing a choice of enumeration follows directly from Tab.~\ref{tab:strongcoloringrules} and minimizes conflicts between paths in a greedy manner. Like the determination of the paths, the outcome of this greedy approach depends on the initial ordering of the interaction set, so, again, it is advantageous to run the algorithm many times.

With $\Pi(\mathcal{P})$ in hand, the algorithm must vertex color it to solve the weak or strong coloring problem, that is to determine the number of steps required to implement the interactions in $\tilde{\mathcal{T}}$. In particular, one can make use of a greedy coloring algorithm as shown in Algorithm~\ref{alg:greedycoloring}. The greedy coloring algorithm is guaranteed to satisfy the bound in Eq.~(\ref{eq:greedybnd}), but its performance can be much better depending on the ordering of vertices. We take a standard approach of a largest-first ordering, where the vertices are ordered from the largest to the smallest degree~\cite{welsh1967upper}. For vertices of the same degree, the ordering is random as determined by the order of the initial randomized interaction list. This largest-first approach often works well in practice, but it is only one option among many~\cite{husfeldt2015graph}.  

\subsection{Analytically Solved Examples Revisited}
We now revisit the analytically solved examples of Sec.~\ref{s:analyticresults}. Using the star graph and complete graph as examples, the algorithm described above is tested for a variety of $N$ ranging from 3 to 35. 1000 different random orderings of the interaction list for each $N$ are considered. As the algorithm is deterministic once a choice of such an ordering is made, this corresponds to 1000 runs of the algorithm. The results are shown in Fig.~\ref{fig:numerics_star_complete}.

For the star graph, the algorithm numerically obtains the true chromatic number for both the weak and strong coloring problem for every ordering of the interaction list. For the weak coloring problem, this is because the conflict graph is well-colored~\footnote{Well-colored graphs are those such that all vertex orderings produce the same number of colors for a greedy coloring.}. In particular, the conflict graph is a co-graph---a class whose members are known to be well-colored~\cite{brandstadt1999graph}. To demonstrate that a graph is a co-graph, it is sufficient to show that it has no length-4 paths as induced subgraphs. Such subgraphs do not exist for $\Pi_\mathrm{weak}(S_N)$. Every vertex $u\in V_{\Pi_\mathrm{weak}}$ is either a member of a complete subgraph $K_{N(N-1)/2}$ consisting of all vertices whose interactions correspond to two-mode interactions or its only neighbors are all contained in such a complete subgraph. Therefore, there exists no set of four vertices whose induced subgraph is a path. While the strong coloring conflict graph is no longer well-colored, the largest-first vertex ordering ensures successful greedy coloring for all interaction-list orderings.

On the other hand, for the complete graph, the algorithm fails to always produce colorings that fully achieve the analytic results. This is because the greedy coloring of the resulting conflict graphs depends heavily on the vertex ordering. Achieving the optimal coloring requires a highly fine-tuned construction. Therefore, the generic randomized greedy coloring algorithm is unlikely to obtain such a coloring as $N$ grows large.
Despite these challenges, the coloring algorithm provides almost optimal results for the complete graph for the graph sizes considered. 

Fortunately for this algorithm, many realistic architectures are expected to result in system graphs that are limited by vertex bottlenecks, given practical limitations on the high qubit connectivity required for producing input/output-limited system graphs like the complete graph.\footnote{Trapped-ion systems are an exception as they provide all-to-all interactions among pairs of qubits.} In the next section, the weak and strong coloring problems will be investigated on two system graphs designed from such realistic architectures.
\begin{figure}
    \centering
    \includegraphics[width=1.0\columnwidth]{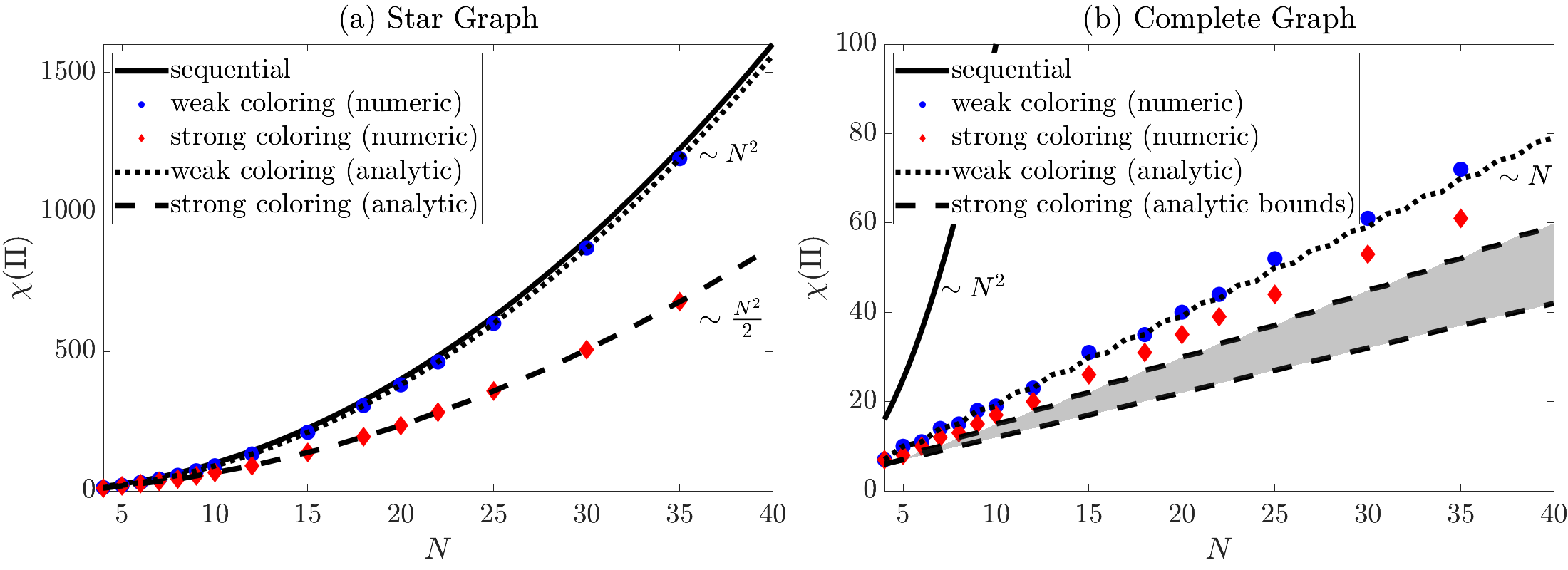}
    \caption{Numerical results for the chromatic number from the weak (blue circles) and strong coloring (red diamonds) problems for (a) star and (b) complete system graphs with $N$ physical fermionic modes. For the non-asymptotic analytic results, see Eqs.~(\ref{eq:weakstar})-(\ref{eq:strongstarNodd}) and Eq.~(\ref{eq:completeweak}), respectively. For the complete graph, the numerical algorithm fails to achieve the analytically determined bounds as obtaining these results requires a highly fine-tuned vertex ordering for the greedy coloring algorithm on the corresponding conflict graph.}
    \label{fig:numerics_star_complete}
\end{figure}

\subsection{Current Architectures}
The algorithm developed in this section can be applied to system graphs designed on examples of realistic superconducting-qubit architectures. Quantum processors built from superconducting qubits have limited connectivity and thus stand the most to gain from optimized parallelizations. The first example to be studied is a heavy-hexagon qubit architecture as used by many of IBM's quantum processors~\cite{IBM}. This architecture has been shown to be favorable for reducing cross-talk and frequency collisions, while allowing for error correction via a hybrid surface and Bacon-Shor code~\cite{chamberland2020topological}. The second example is a square-lattice qubit architecture similar to that used by Google's Sycamore chip~\cite{arute2019quantum}. 

Importantly, qubit architectures are distinct from the system graphs one creates on them. While the qubit architecture places constraints on the design of a system graph, one is free to create many different system graphs on a given quantum processor. In practice, this design problem can be viewed as one of optimizing the limited resources of a particular quantum processor---number of qubits, qubit connectivity, circuit depth---to extract a quality simulation of the largest possible system of fermions. Observe that at the cost of a large circuit depth and high Pauli weight operators, a simple Jordan-Wigner transformation (in the form of a system graph which is a line graph) allows one to simulate the most fermions, as no ancilla qubits are needed. 

One approach to reduce circuit depths and high Pauli-weight operators is to consider more general system graphs. Here, qubit connectivity is a key limitation on designing efficient system graphs if one wants to avoid the need for SWAP operations in the circuit decomposition of the Hamiltonian-simulation algorithm.
In particular, it is desirable to design a system graph so that 1) any qubits that make up a system-graph vertex have linear connectivity for the Jordan-Wigner encoding of the local Majoranas, and 2)
if a pair of vertices are adjacent in the system graph, the internal qubits associated with that edge are adjacent in the architecture graph.  

To apply the algorithm, let us limit ourselves to a single example of a system graph for each qubit architecture under consideration. In particular, in each case, a system with a total of $49$ fermionic modes will be considered. For the heavy hexagon architecture, the system graph considered is constructed from 65 qubits and mirrors the structure of the underlying qubits. This is identical to an example considered in Ref.~\cite{chien2020custom}. For the square lattice, the system graph considered is a triangular tiling of the Euclidean plane and requires 147 qubits. The precise mappings from architecture graphs to system graphs for each of these cases are shown in Fig.~\ref{fig:architecturetosystem}.
\begin{figure}
    \centering
    \includegraphics[width=0.95\columnwidth]{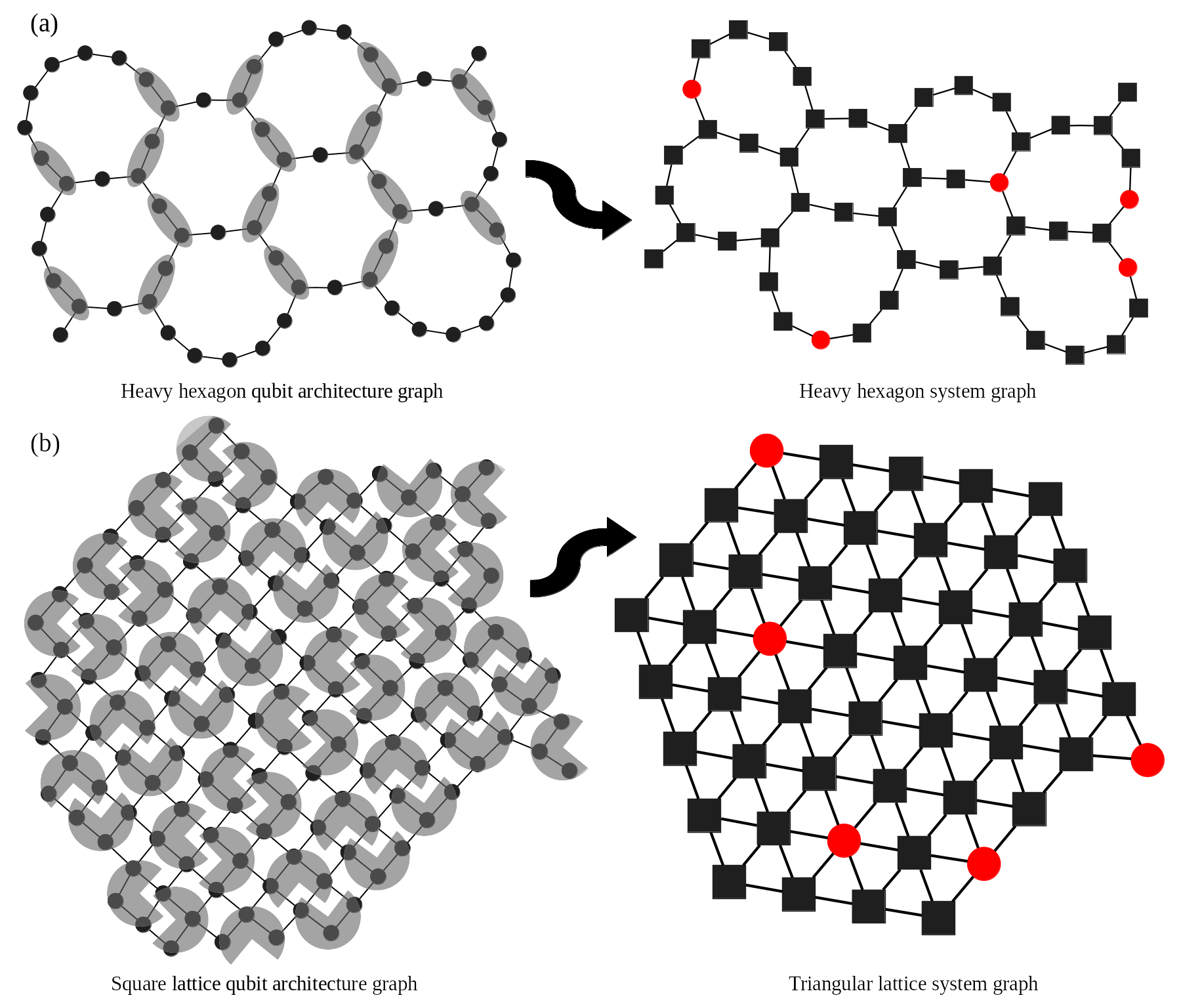}
    \caption{Mappings from architecture graphs to system graphs. On the architecture graph, qubits are represented by black dots and are grouped into system-graph vertices as denoted by the gray shading. On the system graphs, square vertices denote virtual vertices and red circular vertices denote physical vertices. For the numerics, $N$ physical vertices are chosen randomly from the 49 total fermionic modes.  }.
    \label{fig:architecturetosystem}
\end{figure}

Given these system graphs, $N$ of the $49$ vertices are randomly selected to be physical vertices for various $N$ between 5 and 35. For each $N$, 50 random choices of physical vertices are considered and on each instance, the algorithm is run for 1000 different random orderings of the interaction set for both the weak and strong coloring problems. The best solution from these 1000 different random orderings is then taken. These results, along with quadratic fits are shown in Fig.~\ref{fig:numerics_architectures}.
\begin{figure}
    \centering
    \includegraphics[width=0.95\columnwidth]{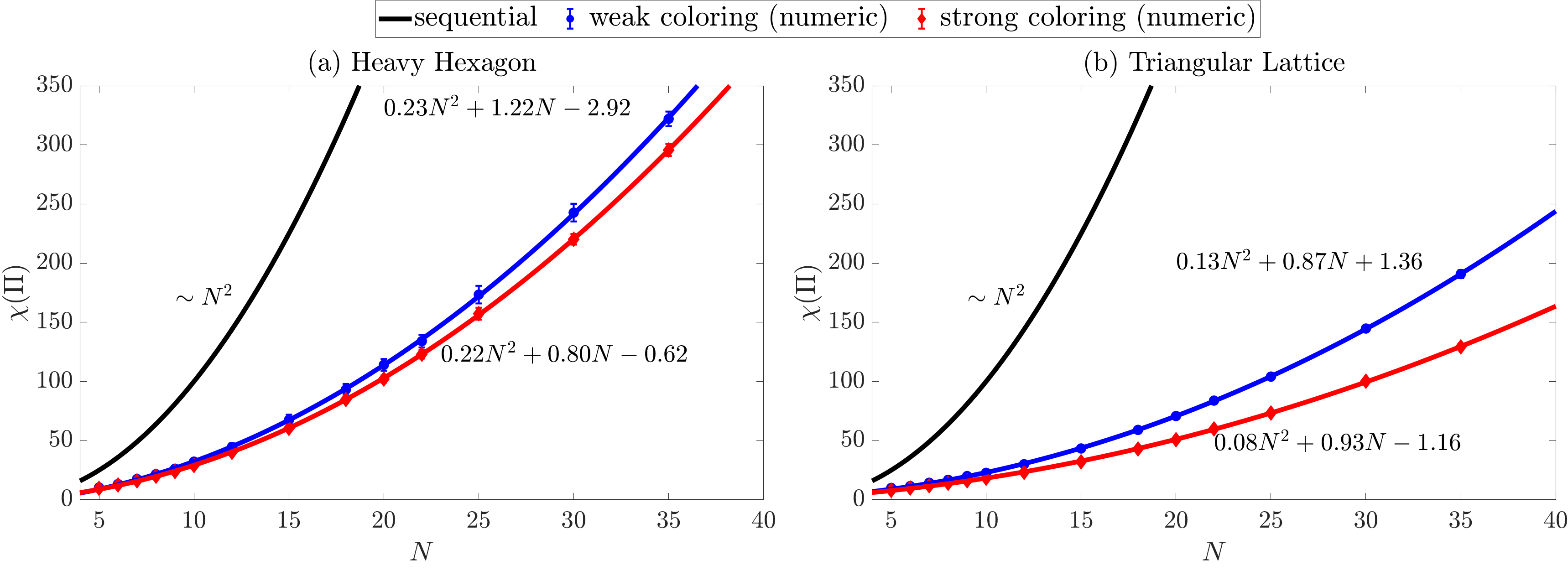}
    \caption{Numerical results for the chromatic number of the conflict graph as a function of the number of qubits $N$ for (a) the heavy-hexagon system graph and (b) the triangular-lattice system graph compared with the scaling of the sequential implementation of the Hamiltonian terms.}.
    \label{fig:numerics_architectures}
\end{figure}

As is seen from the plots, the improvement from strong coloring over weak coloring in the case of the heavy-hexagon graph is minimal compared to the improvement in the triangular lattice. This is consistent with the conclusions of Sect.~\ref{s:analyticresults}: strong coloring provides the higher performance benefit when the size of the edge bottlenecks to routing the paths induced by interactions are much larger than the size of vertex bottlenecks. The triangular lattice has many more edges per vertex (and correspondingly more qubits) which enable greater parallelization via strong coloring.

Figure~\ref{fig:numerics_ster} shows tradeoffs between the number of qubits and the degree of parallelization for the various examples considered in this work: the complete graph, the star graph, the heavy-hexagon graph, and the triangular lattice. While the best balancing of these various tradeoffs depends on many variables, the triangular lattice serves as a particularly nice example of how a system graph on a realistic architecture can be subject to significant reductions in circuit depth via parallelization. While weak coloring alone offers significant performance gains over a sequential approach, considering the precise details of this mapping via the strong coloring problem is important for minimizing the circuit depth. 
\begin{figure}
    \centering
    \includegraphics[width=\columnwidth]{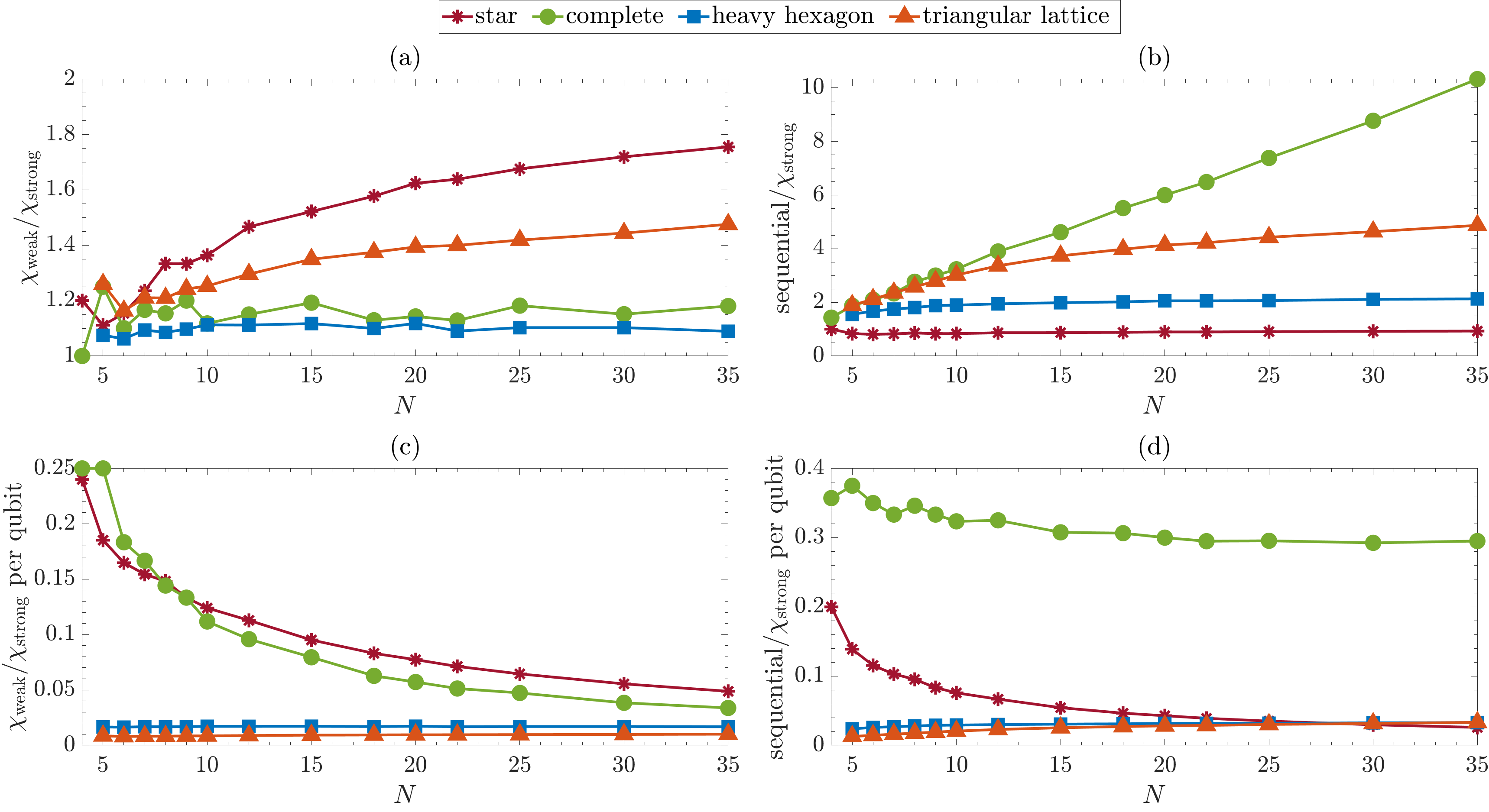}
    \caption{Numerical results on the amount of parallelization for various system graphs. (a) shows the improvement in the number of steps for strong coloring versus weak coloring. (b) shows the improvement for strong coloring over a sequential implementation of the interactions in the interaction list. (c) and (d) show the same as (a) and (b), respectively, but weighted by the number of qubits in the system graph. When comparing the different system graphs, recall the size of the complete graph and star graph scale with the number of physical vertices $N$ whereas the heavy hexagon and triangular lattice do not.}
    \label{fig:numerics_ster}
\end{figure}

\subsection{Local Interactions}\label{sec:NN-Hamiltonian}
A more common fermionic interaction term is nearest-neighbor hopping on a lattice. The parallelization of this work within the framework of custom fermionic codes can be applied to optimize simulating this model as well. Consider an interaction graph in the form of a two-dimensional square lattice of physical fermionic modes with nearest-neighbor hopping and open boundary conditions,
\begin{equation}
    H = \sum_{\langle u,v\rangle} \kappa_{uv} a^\dagger_u a_v + \sum_u \kappa_{uu} a^\dagger_u a_u,
\end{equation}
for real $\kappa_{uv}=\kappa_{vu}$, where the sum is over neighbors on the square lattice. Minimizing the Pauli weight and maximizing the parallelization of such nearest-neighbor hopping terms is the limiting algorithmic factor for a variety of models of interest, such as the spinless Fermi-Hubbard model on the square lattice. This problem is well-understood analytically for a variety of specific fermion-to-qubit mappings~\cite{verstraete2005mapping, steudtner2019quantum, derby2021compact}. While such analytic approaches to specific problems are valuable when tractable, the techniques of this work allow for an automated optimization for arbitrary Hamiltonians and arbitrary system graphs. 

We consider this problem for three different system graphs. The first case is a system graph identical to the interaction graph---a two-dimensional square lattice with all physical fermionic modes. This case can be directly compared with previous work on this problem. The other two cases involve placing the physical fermionic modes in the heavy-hexagon and triangular-lattice system graphs considered above (see Fig.~\ref{fig:architecturetosystem}). The physical modes are embedded such that nearest neighbors on the square-lattice interaction graph are as close as possible on the system graph, enabling low-weight Pauli strings. The precise details of this mapping are included with the source code as supplemental material~\cite{github}.

The results for weak and strong coloring on each of these system graphs for a range of interaction-graph lattice sizes are shown in Fig.~\ref{fig:squarelattice}. For each lattice size $L\times L$ for $L\in\{4,9,16,25,36\}$, 1000 random orders for the interaction list are considered for both weak and strong coloring. Observe that compared to non-local models, the advantage over weak coloring due to strong coloring is minimal, independent of system graph. This is expected since all path lengths are short in this model and, similar to the complete-graph example, parallelization is limited by the input/output capacity of the physical vertices. However, parallelization provides significant gains over a naive sequential strategy which scales as $\sim 4N$.

In agreement with previous work, an $O(1)$ circuit depth is obtained with increasing lattice size for the square-lattice interaction graph, and the local fermionic interactions are mapped to local qubit interactions. The triangular lattice performs similarly although it allows for slightly improved performance, especially when small numbers of physical fermions are embedded in the system graph. This is because the triangular lattice offers more paths for implementing interactions than the square lattice. The spike in chromatic number for $N=36$ in the triangular lattice is because a $6\times 6$ square lattice cannot quite fit in the triangular-lattice system graph considered. Therefore, some interactions that are local in the interaction graph become non-local in the system graph. This effect is even more pronounced for the heavy-hexagon system graph which has lower connectivity than the interaction graph, and therefore cannot perform as well as the other system graphs even for small system sizes. 
\begin{figure}
    \centering
    \includegraphics[scale=0.28]{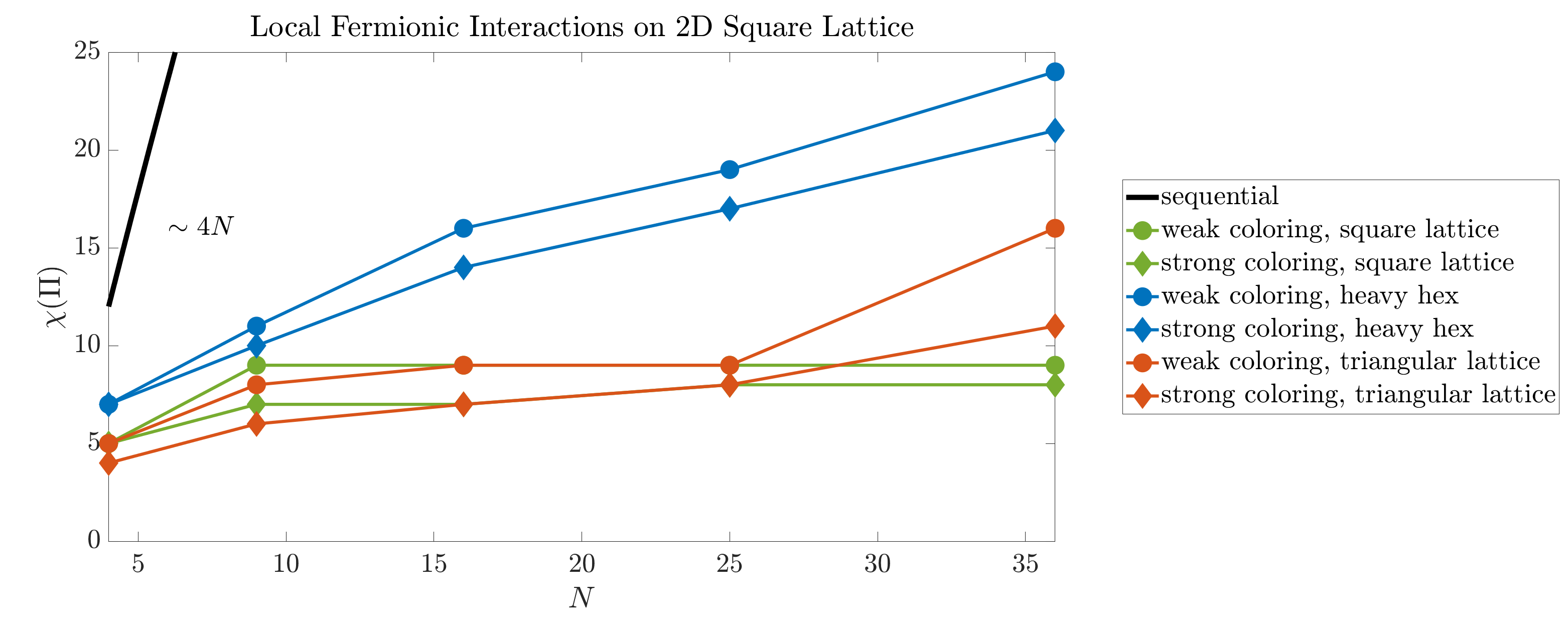}
    \caption{Numerical results for the chromatic number of the conflict graph of the system graphs noted corresponding to a square lattice with nearest-neighbor hopping for weak and strong coloring problems.}
    \label{fig:squarelattice}
\end{figure}

\section{Conclusion and Outlook}\label{s:conclusion}
The amount of parallelization afforded by a system graph is an important target for optimization in the quantum simulation of fermionic Hamiltonians on near-term quantum processors where circuit depth is expected to be an important limiting factor. In this work, this problem is mapped to a graph coloring problem and the relationship between parallelization and the system-graph structure for a variety of representative examples are explored both analytically and numerically. It is found that by considering the details of the fermion-to-qubit mapping, that is to seek strong coloring, one can often find constant-factor improvements in parallelizability relative to performing only weak coloring which is a more high-level approach. Both approaches enable significant reductions in circuit depth relative to a naive sequential approach. The amount of improvement of both coloring schemes compared with the sequential approach, and the strong versus weak coloring is a function of system-graph characteristics---for instance, the number of and the severity of system graph's vertex and edge bottlenecks---as well as on the choice of enumerating edges in the system graph.

A full account of the algorithmic costs for a Hamiltonian of interest would incorporate the algorithms for parallelization presented here to design a fermion-to-qubit mapping that respects hardware-specific constraints, such as qubit connectivity, noise tolerance, and implementable circuit depths. This work considers only one approach to parallelizability offered at the level of the number of steps needed to implement the Pauli strings that result from a custom fermionic code. When attempting to fully optimize a simulation algorithm in an architecture-aware manner, our approach may further be combined with other parallelization schemes, e.g., those based on fermionic SWAP networks~\cite{babbush2017low, kivlichan2018quantum, cade2020strategies} or approaches that concern fine-grained details of the circuit decomposition when the Pauli strings are compiled to basic two-qubit entangling gates~\cite{hastings2014improving}. 
The strong-coloring problem, in particular, depends heavily on the choice of encoding of the local Majoranas. While this work only considers Jordan-Wigner encoding of the local Majoranas, it is known that other choices (e.g. Fenwick-tree encoding~\cite{bravyi2002fermionic,havlicek2017operator,setia2019superfast}) lead to lower Pauli weights for the local operators---potentially at the cost of reducing the possibility of parallelization via strong coloring. 
In fact, local Majoranas could be encoded differently on different sites to perform a full optimization at the circuit level. The problem of detailing the strong coloring rules for other (possibly mixed) choices of encoding the local Majoranas is left to for future studies. 

The custom fermionic codes considered in this work, and the generalizations described above, encompass a broad range of mappings.
Nonetheless, these do not exhaust the possibilities for mapping fermions to qubits. 
Consequently, one can imagine profitably mapping parallelization tasks to graph coloring for other encodings as well. 
For instance, while weak coloring allows one to parallelize a Jordan-Wigner encoding (contained in the class of encodings of this work as a system graph consisting of a line of vertices), another ancilla-free mapping, the Bravyi-Kitaev encoding~\cite{bravyi2002fermionic}, does not allow for this sort of improvement. This is because the structure of the Bravyi-Kitaev encoding is given by a Fenwick tree~\cite{havlicek2017operator}, where the root qubit of the tree is non-trivially acted on for every operation, preventing parallelization of the sort we consider here.

We anticipate that applying our tool-set for analyzing parallelizability for Hamiltonian simulation in conjunction with architectural considerations will be useful for obtaining detailed simulation costs for other fermionic Hamiltonians not studied in this work. For example, local and non-local interactions involving four fermionic operators (e.g., Coulomb interactions in quantum chemistry and two-nucleon interactions in nuclear physics) and interactions involving more fermions (such as three- and higher-body interactions in nuclear physics~\cite{machleidt2011chiral}) can be incorporated in the parallelization scheme of this work, and lead to improved simulations in the near and far term. In another interesting direction, one may consider applying the strategy of this work in designing parallelized simulation steps in connection to system graph and hardware connectivity to interacting systems of fermions and bosons, such as those of relevance to lattice gauge theories~\cite{byrnes2006simulating,lamm2019general,shaw2020quantum,paulson2021simulating,ciavarella2021trailhead,kan2021lattice}. For example, it would be interesting to thoroughly examine the simulation cost, considering parallelization potential, of fully fermionic formulations (that can be achieved only in 1+1 dimensions~\cite{hamer1997series}) and fully bosonic formulations (that can be achieved for certain gauge theories~\cite{zohar2019removing}) of a lattice gauge theory~\cite{davoudi2021search}. Finally, in designing system graphs, one may need to take into consideration the entanglement structure (see e.g., Ref.~\cite{klco2021geometric} for a discussion in the context of quantum fields) of the resulting subgraphs and the associated computational complexity of various simulation steps, such as state preparation, that is closely tied to entanglement properties.

\begin{acknowledgements}
We thank  Aniruddha Bapat,  Yu-An Chen, Michael Jarret, Alexander Schuckert, and James Watson for valuable discussions. J.B.~acknowledges support by the U.S.~Department of Energy, Office of Science, Office of Advanced Scientific Computing Research, Department of Energy Computational Science Graduate Fellowship under award No.~DE-SC0019323 and funding by the DoE ASCR Accelerated Research in Quantum Computing program (award No. DE-SC0020312), DoE QSA, NSF QLCI (award No. OMA-2120757), DoE ASCR Quantum Testbed Pathfinder program (award No. DE-SC0019040), NSF PFCQC program, AFOSR, ARO MURI, AFOSR MURI, and DARPA SAVaNT ADVENT. Z.D. acknowledges support by the U.S. Department of Energy, Office of Science, Early
Career Award under award no. DE-SC0020271. This research was supported in part by the National Science Foundation under Grant No. NSF PHY-1748958, the Heising-Simons Foundation, and the Simons Foundation (216179, LB).
\end{acknowledgements}

\bibliography{main.bib}

\end{document}